%% file: main.tex
\documentclass[journal,draftcls,onecolumn,12pt,twoside]{IEEEtran}

\ifCLASSINFOpdf
\fi
\usepackage{verbatim} 
\usepackage{multirow}
\usepackage[dvips]{graphics} 
\usepackage{graphicx} 
\usepackage{times} 
\usepackage[font=small,labelfont=bf]{caption} 
\usepackage{epsfig} 
\usepackage{times} 
\usepackage{caption} 
\usepackage[font=small,labelfont=bf]{caption}
\usepackage{epsfig}
\usepackage{stfloats}
\usepackage{float} 
\usepackage[utf8]{inputenc}
\usepackage[english]{babel}
\usepackage{hyperref,xcolor}
\usepackage[noadjust]{cite}
\usepackage{enumerate}
\usepackage{amssymb}
\usepackage{amsthm}
\usepackage{multirow}
\usepackage{latexsym}
\usepackage{color}
\usepackage{float}
\usepackage{amsmath}
\usepackage{cite} 
\usepackage[noend]{algpseudocode}
\usepackage{epstopdf}
\usepackage{tikz}
\usepackage{threeparttable}
\usepackage{booktabs}
\usepackage{algorithm}
\usepackage{mathtools}

\allowdisplaybreaks
\usepackage{etoolbox}
\newtheorem{theorem}{Theorem} 
\newtheorem{remark}{Remark}
\makeatletter
\patchcmd{\@begintheorem}{\textit}{\textbf}{}{}

\newtheorem{lemma}{Lemma}

\usepackage{tikz}
\usepackage{pgfplots}
\usepackage{tikzscale}
\usepackage{tikz-qtree}

\newcommand*{\algrule}[1][\algorithmicindent]{%
 \makebox[#1][l]{%
 \hspace*{.2em}
 \vrule height .75\baselineskip depth .25\baselineskip
 }
}

\newcount\ALG@printindent@tempcnta
\def\ALG@printindent{%
 \ifnum \theALG@nested>0
 \ifx\ALG@text\ALG@x@notext
 \else
 \unskip
 \ALG@printindent@tempcnta=1
 \loop
 \algrule[\csname ALG@ind@\the\ALG@printindent@tempcnta\endcsname]%
 \advance \ALG@printindent@tempcnta 1
 \ifnum \ALG@printindent@tempcnta<\numexpr\theALG@nested+1\relax
 \repeat
 \fi
 \fi
}
\patchcmd{\ALG@doentity}{\noindent\hskip\ALG@tlm}{\ALG@printindent}{}{\errmessage{failed to patch}}
\patchcmd{\ALG@doentity}{\item[]\nointerlineskip}{}{}{} 
\makeatother 
\usepackage[T1]{fontenc}
\usepackage{amsmath}

\usepackage{amsmath}

\newcommand{\FC}{\textcolor{blue}}
\newcommand{\CS}[1]{\textcolor{magenta}{#1}}

\usepackage{amsmath}

\usepackage{subcaption}

\def \fwidth{0.7\linewidth}
\def \fheight {0.3\linewidth}

\def \sfwidth{0.99\linewidth}
\def \sfheight {0.6\linewidth}

\definecolor{color0}{HTML}{FFD700}
\definecolor{color1}{HTML}{FFB14E}
\definecolor{color2}{HTML}{FA8775}
\definecolor{color3}{HTML}{EA5F94}
\definecolor{color4}{HTML}{CD34B5}
\definecolor{color5}{HTML}{9D02D7}
\definecolor{color6}{HTML}{0000FF}

\pgfplotsset{every axis/.append style={
                    label style={font=\scriptsize},
                    tick label style={font=\scriptsize}  
                    }}

\pgfplotsset{compat=1.17}
                    
\begin{document}


\title{Statistical Characterization of \\ Closed-Loop Latency at the Mobile Edge} 

\author{ Suraj Suman,~\IEEEmembership{Member, IEEE}, Federico Chiariotti,~\IEEEmembership{Member, IEEE}, \v Cedomir Stefanovi\' c,~\IEEEmembership{Senior Member, IEEE}, Strahinja Do\v sen,~\IEEEmembership{Member, IEEE}, and Petar Popovski,~\IEEEmembership{Fellow, IEEE}
\thanks{ Suraj Suman, Federico Chiariotti, \v Cedomir Stefanovi\' c, and Petar Popovski are with Department of Electronic Systems, Aalborg University, Denmark. Strahinja Do\v sen is with Department of Health Science and Technology, Aalborg University, Denmark. }}

\maketitle

\vspace{-2.15cm} 


\begin{abstract}
The stringent timing and reliability requirements in mission-critical applications require a detailed statistical characterization of the latency. Teleoperation is a representative use case, in which a human operator (HO) remotely controls a robot by exchanging command and feedback signals. We present a framework to analyze the latency of a closed-loop teleoperation system consisting of three entities: HO, robot located in remote environment, and a Base Station (BS) with Mobile edge Computing (MEC) capabilities. A model of each component of the system is used to analyze the closed-loop latency and decide upon the optimal compression strategy. 
The closed-form expression of the distribution of the closed-loop latency is difficult to estimate, such that suitable upper and lower bounds are obtained. We formulate a non-convex optimization problem to minimize the closed-loop latency. Using the obtained upper and lower bound on the closed-loop latency, a computationally efficient procedure to optimize the closed-loop latency is presented. The simulation results reveal that compression of sensing data is not always beneficial, while system design based on average performance leads to under-provisioning and may cause performance degradation. The applicability of the proposed analysis is much wider than teleoperation, for systems whose latency budget consists of many components.
\end{abstract}



\section{Introduction} 

The Tactile Internet (TI)~\cite{Fettweis_TI} is a fairly recent concept that involves the transmission of tactile sensations along with data, text, and multimedia. 
The ability to receive tactile stimulation enhances the immersion of the user in Virtual Reality (VR), and improves control performance in teleoperation, in which a human operator (HO) controls and manipulates a remotely located robot or object~\cite{EMG_teleoperation_1}. 
This involves a two-way exchange of data: commands from the operator and tactile and feedback from the remote environment, creating a closed-loop system \cite{antonakoglou2018toward}.
Advanced Human-to-Machine Interaction (HMI) in teleoperation involves the exchange of abundant sensory data (including textual, visual, and tactile information), which ensures that the HO can have an intuitive and precise interaction with the remote environment, improving the task execution accuracy and efficiency, but also putting a significant strain on the communication system.

Due to its interactive nature, teleoperation is highly sensitive to communication impairments.
The concept of motion-to-photon delay~\cite{zhao2017estimating}, often used in VR, is extremely relevant here: if we measure the closed-loop latency between the moment an action is performed by the operator and the moment that they get the related feedback, we can gauge the Quality of Experience (QoE) for the operator, as well as the final control performance.
Besides low latency, teleoperation requires a high reliability, as the communication channel losses may degrade the control precision.
Moreover, reliable low-latency systems are effectively~\emph{transparent}~\cite{transparency_work}, i.e., the operator does not notice the remote nature of the environment, and feels as if they were controlling the robot directly.
This ensures the immersion of the operator in the remote environment, improving both QoE and control performance.
On the other hand, long delays and high jitter can both deteriorate the user experience and jeopardize the stability of the
closed-loop teleoperation system, as the reactions to the events in the remote environment are delayed and  sluggish~\cite{stability_td_hri}.  {{For example, comfort zone for performing telesurgery is with round trip latency of around $300$ ms \cite{barba2022remote}, whereas it is around $800$ ms in VR-based haptic teleoperation \cite{valenzuela2019virtual}.  }}


Communication is not the only bottleneck in the system, as closed-loop latency also includes computation and processing times, while the robot might have limited on-board computation capabilities.
As the decoding and execution of the operator's commands along with the compression and processing of the sensing data are part of the teleoperation control loop, these operations can have a significant effect on the latency.
In effect, the success of teleoperation applications strongly depends on the performance of both the communication and the computational segments. In order to fully optimize the teleoperation system to meet the closed-loop latency constraints required by the application, which can be very stringent in mission-critical industrial scenarios, we need to consider all steps of the process.

Latency in networked, closed-loop control systems is not easy to control, as the randomness of the wireless channel and the variable amount of available bandwidth and computation resources make the closed-loop latency a stochastic quantity.
While the effects of a higher latency and jitter are understood from laboratory experiments~\cite{effect_of_delay, low_latency_h2m, closing_the_force_loop}, existing schemes optimize only for the \emph{average} latency~\cite{IoT_TI, NFV_e2e_QoS,onu_2}, without providing any reliability guarantees.
Furthermore, most teleoperation system designs~\cite{MEC_tactility, MEC_TII, hybrid_caching_TI, TelSurg} do not take into account the computational delay into the latency calculation, effectively leaving out part of the control loop and potentially disregarding an important component of the closed-loop latency.

\begin{figure}[t!]
    \centering
    \input{tikz_figs/sysmodel}
\caption{System model for closed-loop MEC-enabled teleoperation system.} 
\label{fig:sys_mod_MEC_TI} 
\end{figure}
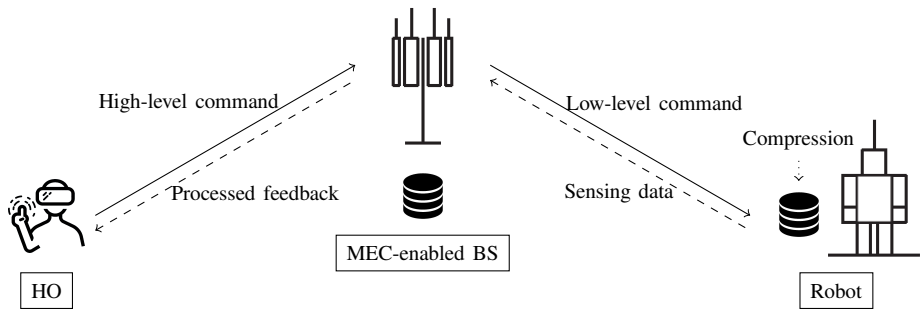

In this work, we present an analysis of a Mobile Edge Computing (MEC)~\cite{MEC_survey, MEC_survey_IoTJ} system for teleoperation, in which the most computationally-intensive tasks are offloaded to the MEC-enabled Base Station (BS) handles  as shown in Fig. \ref{fig:sys_mod_MEC_TI}. A preliminary version of this work was presented in~\cite{our_icc}, focusing only on the latency of an uplink connection, rather than a closed-loop latency, from a multi-sensor robot to the BS. This work generalizes the framework to analyze the latency of the closed-loop teleoperation system and formulate an optimization problem to minimize the closed-loop latency with high reliability under statistical constraints. The key contributions of the paper are as follows: 

\begin{enumerate}
\item The system model for closed-loop MEC-enabled teleoperation system is presented, where the command and sensing data are exchanged over wireless network through a BS. A MEC server on the BS processes the command and sensing data in order to be sent to the robot and the feedback for the HO, respectively, enabling the closed-loop operation.

\item Two system design possibilities are considered and compared. In the first one, the robot located in the remote environment compresses the raw sensing data first, then transmits it to the BS, which decompresses and processes it to extract the user readable feedback that is transmitted back to the HO. In the second scenario, the robot transmits the raw sensing data to the BS without compression, where it is processed by the MEC server. 

\item The closed-loop latency for both scenarios is analyzed by dividing it into three components: the duration of the compression operation on the sensing data, the transmission delays for commands and feedback, and the decompression and computation at the MEC server. All these latency components are modelled as independent random variables (RVs). Thus, the closed-loop latency, which is the sum of all these components, is also an RV. We characterize the nature of its Probability Density Function (PDF) and Cumulative Distribution Function (CDF) and obtain tractable upper and lower bounds on the closed-loop latency for both scenarios. 

\item A non-convex optimization problem is formulated, aiming to minimize the the closed-loop latency in the statistical sense. 
Using the obtained upper and lower bounds on the closed-loop latency, we present a computationally efficient procedure to estimate the optimal closed-loop latency and solve the problem.

\item We analyze the performance of the schemes by simulation, and find that the simulation results reveal that compression of sensing data is not always beneficial. The decision on data compression depends on the communication system parameters, as well as the computational capability of the robot. The proposed approach is also compared with the system design that is optimized in average sense, as reported in the prior works. The comparative analysis reveals that system design in average sense leads to under-provisioning and causes a significant performance degradation.  
\end{enumerate}
While this work focuses on the latency analysis for teleoperation, the basic framework we propose can be used for any cascaded system with random latency components, such as  Open Radio Access Network (O-RAN) \cite{o-ran} systems. The O-RAN architecture is envisioned to execute networking processes in software, making network components' behavior programmable. Telecom operators will use the standardized interfaces to control multi-vendor infrastructures. In the context of O-RAN, the proposed framework will be very useful to analyze the latency incurred across multiple software and hardware components from multiple vendors in order to maximize user QoE.

The rest of this paper is organized as follows. Section~\ref{sec:related} presents related work on the subject, while the basic model of a MEC-enabled teleoperation system is presented in Section~\ref{sec:system}. The different delay components of the closed-loop latency and their distributions are discussed in Section~\ref{sec:latency}, and the overall closed-loop latency distribution is estimated and optimized in Section~\ref{sec:analysis}. Finally, the numerical simulation results are discussed in Section~\ref{sec:results}, followed by the concluding remarks in Section~\ref{sec:conc}. 

\section{Related Work}\label{sec:related}

Latency is a major issue in teleoperation systems, and several studies \cite{effect_of_delay, low_latency_h2m, closing_the_force_loop} have investigated its impact on control performance. The study in \cite{effect_of_delay} describes an experiment that uses a haptic device to generate feedback, presenting the visual three-dimensional environment to the user on a monitor and studying the effect of latency between the participant’s actual action and the visible movement on the monitor. A commercial haptic teleoperation system \cite{cyber_glove} was used in \cite{low_latency_h2m}, which allowed HOs to touch and grasp the computer-generated virtual objects. This experiment demonstrated that the average latency increases significantly with the network load. A closed-loop compensatory tracking task is performed in the closed-loop control using tactile input in \cite{s_dosen_ET_fb}, where the feedback is encoded to the user using frequency and amplitude modulation schemes. Significant time delay on the order of several hundred milliseconds has been noticed in this experiment.

Most existing low-latency teleoperation schemes have tried to minimize the \emph{average} latency, which is an easier target, and neglect the required reliability targets in statistical sense. Even with fiber-wireless (FiWi) networks, existing optimization works are limited to average guarantees~\cite{onu_2}. 
In this scenario, the limiting factors are the availability, skill set, distance to task location, and remaining energy of robots~\cite{onu_1}, or the association between tasks and HOs~\cite{onu_3}. The study in \cite{coll_computing} presents a task allocation strategy by combining suitable host robot selection and computation task offloading onto collaborative nodes in the FiWi infrastructure. The conventional Cloud, decentralized cloudlets, and neighboring robots as collaborative nodes are used for computation offloading. 
Cross-layer techniques for low-latency teleoperation have also been considered in the literature~\cite{cross_layer_3, TI_channel_access, cross_layer_1}. TI cross-layer transmission optimization is investigated in~\cite{cross_layer_3} by considering the transmission delay, error probability and statistical queuing delay requirements, using a proactive packet dropping mechanism to limit latency. A resource allocation mechanism to maximize the uplink sum rate of traditional data while satisfying the delay requirements for tactile data is presented in \cite{cross_layer_1} using sparse code multiple access. The study in \cite{TI_channel_access} estimates the average latency from a hub to an access point for tactile body-worn devices connected using an IEEE 802.11 network.

In general, support for teleoperation applications based on Network Function Virtualization (NFV) is included in the 5G network architecture~\cite{NFV_JSAC, IoT_TI, NFV_e2e_QoS}. The study in~\cite{IoT_TI} presents a utility function based model to evaluate the performance of the NFV-based TI by considering the human perception resolution and the network cost of completing services. The utility function depends on the average round-trip delay, network link bandwidth, and node virtual resource consumption. The joint radio and NFV resources for a heterogeneous network are allocated in \cite{NFV_e2e_QoS} by guaranteeing average end-to-end delay of each tactile user, including the queuing, transmission, and computation delays. MEC offloading is another possibility for TI applications~\cite{TelSurg,MEC_TII}. A trade-off between the average service response time and power usage efficiency is investigated in \cite{MEC_tactility} for local and cooperative MEC. This can be optimized according to QoE metrics as well~\cite{MEC_TII}, or using more advanced caching techniques~\cite{hybrid_caching_TI}. Finally, a real-time network architecture for remote surgery application is presented in \cite{TelSurg}, employing cloud and MEC networks to satisfy the timing constraints, which are still expressed in terms of the average end-to-end delay.

Overall, existing experimental works have mostly considered simple links with controllable latency, while existing architectures only dealt with average latency, and often disregarded the contribution of the computational component of the closed-loop latency.
The main novelty of our work is in the consideration of these factors, estimating the complete distribution of the closed-loop latency in a public Internet setup along with optimizing it for the worst-case scenario by considering statistical guarantees to ensure a more stable performance than average latency minimization.

\section{System Model} \label{sec:system}
The model of the considered MEC-enabled teleoperation system is shown in Fig. \ref{fig:sys_mod_MEC_TI}: the robot, which is located in the remote environment, is instructed by the HO, who receives feedback information that guides his decisions. 
The instruction sent to the robot is denoted as the \emph{command signal}, whereas the the information received from the robot as the \emph{sensing data}. The command signal from the HO and the sensing data from the robot are exchanged over a wireless connection to the BS. The volume of sensing data generated is much larger than the command signal, because it can include throughput-intensive formats such as video and tactile data, along with other types of data such as audio, image, or text. The sensing data acts as feedback to the HO for further command instructions to the remote robot. The transmission of this potentially large volume of data over a wireless connection is expensive in terms of required radio resources, and may require compression before the data is transmitted. 

The alternative is to process the sensing data at the robot itself and extract the user-readable feedback signal, but this may be very computational-intensive and even far beyond the capacity of the robot. On the other hand, commands from the HO are typically not compressed, and represent high-level instructions to the robot. The MEC server then translates these high-level commands to low-level commands to the robot's actuators, which can be directly executed. As the size of command signals is typically much smaller than the sensing data from the robot (commands usually come from a limited set), the compression of command data is not within the scope of our work  \footnote{Note that the analytical framework presented here can be straightforwardly extended if the command data also gets compressed at the HO side before transmission.}.

The MEC server equipped at the BS acts as a decision-support system that handles the data-intensive computation task by processing the sensing data from the robot. 
Such processed data is communicated to the HO. 
In the system model shown in Fig. \ref{fig:sys_mod_MEC_TI}, the robot compresses the sensing data locally and transmits this compressed data to the BS. The MEC server at the BS first decompresses the data, processes it, and sends the processed data having user readable command to the HO. Based on the received feedback, the HO decides on the command signal for the robot located in the remote environment.
Thus, the command and sensing signals exchanged over wireless medium through a BS form a closed-loop teleoperation system connected over two communication links. 
 
\begin{remark} 
 The focus of this work is on communication and computation aspects of the telemanupulation system. Therefore, the latency due to executing the command signal at the robot is not taken into account, because its value is negligible as compared to the latency components shown in Fig.~\ref{fig:depiction_delay}. Further, it depends upon the mechanical property of the robot and application at hand. Apart from it, the latency incurred in the HO's reaction is not considered, as this component is beyond the scope of the work. 
\end{remark}



\begin{remark}
In the following, we assume that the data related to teleoperation application is allocated resources in a private slice by the BS avoiding queuing delay, as its latency-constrained nature requires network support.
\end{remark}



\begin{figure}[t!]
    \centering
    \input{tikz_figs/latency_components}
\caption{Depiction of different components of the closed-loop latency.} 
\label{fig:depiction_delay} 
\end{figure}
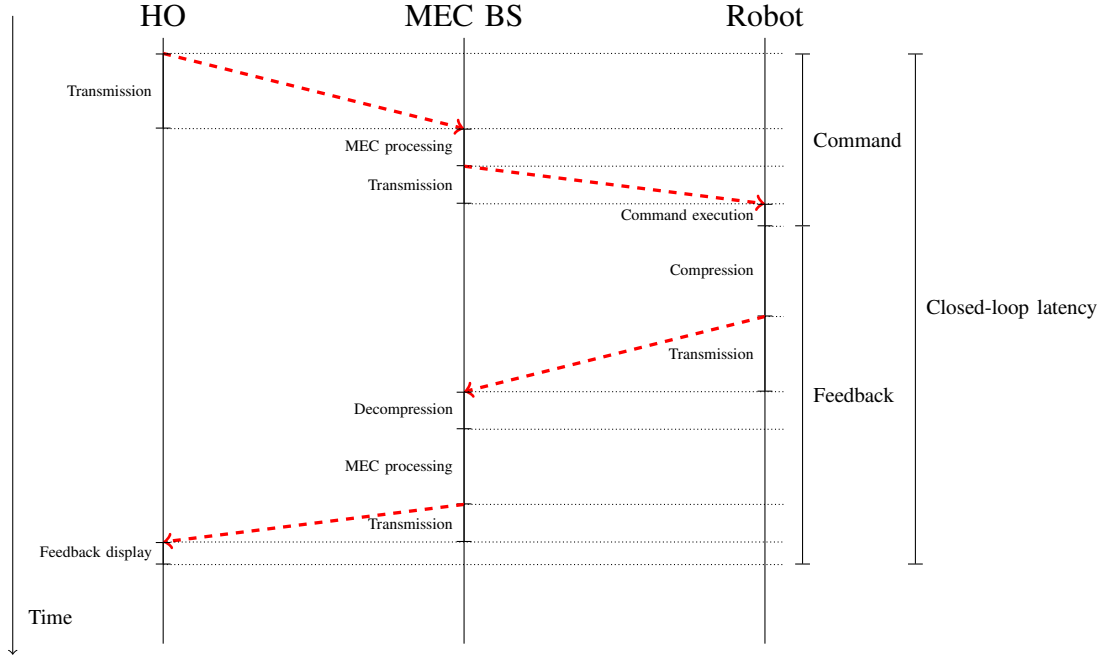

The different delay components of the closed-loop latency, which constitute the delay incurred in exchanging command and sensing signals, are shown in Fig. \ref{fig:depiction_delay} for the system model in Fig.~\ref{fig:sys_mod_MEC_TI}. 
The feed-forward delay consists of the delay incurred in the transmission of the high-level command signal, the processing at MEC (if required), and the transmission of processed data (low-level command). 
 The feedback delay consists of the delay incurred in compressing the sensing data at the robot, the transmission of these compressed data, the decompression and processing at the MEC BS, and finally, in the transmission of this processed user-readable feedback signal to the HO
 \footnote{The propagation delay 
 is not taken into consideration here, as it will be several orders of magnitude lower than the other components of the closed-loop delay for the distances relevant for the considered teleoperation scenario. 
 }.

\section{Latency Components}\label{sec:latency}
In the given context of a MEC-enabled teleoperation system, the closed-loop latency is mainly due to the delay incurred in transmission of command and sensing data, compression at the robot, and decompression and computation at the MEC server. The data transmission delay is random in nature due to uncertainty in the wireless channel, while the randomness of the computation time at the MEC server is related to the uncertainty in the amount of resources allocated for processing. Thus, the overall closed-loop latency is also an RV. In the following, we characterize the probability distributions of different delay components of the closed-loop MEC-enabled teleoperation system.

\subsection{Latency  Incurred in Data Transmission}
We consider a wireless channel with Rayleigh block fading, such that the channel gain $h$ is constant over the length of the packet. Hence, the fading gain $g=|h|^2\sim \exp(1)$, and the gain over subsequent packets is independent and identically distributed.
We consider a transmitter that uses a constant power $P_\text{tx}$ located at a distance $d$ from the receiver. 
As teleoperation applications are extremely sensitive to delay and require dedicated resources, we consider a slicing-enabled 5G system in which a bandwidth $B$ is reserved for the transmission~\cite{aijaz2017shaping}. The signal-to-noise ratio (SNR) $\gamma$ at the receiver is then:
\begin{align} \label{eq:snr}
\gamma(P_\text{tx},d, B) & = K_0 \frac{P_\text{tx} |h|^2}{ d^{\ell} N_0 B} = \gamma_0(P_\text{tx},d, B) \cdot g,
\end{align}
where $K_0$ is the Friss equation parameter, $\ell$ is a path loss exponent depending on the propagation scenario, $N_0$ is the noise power spectral density, and $\gamma_0(P_\text{tx}, d, B) = K_0 \frac{P_\text{tx} }{ d^{\ell} N_0 B}$ is the average SNR. 




We assume that the transmitter can choose the transmission rate, guaranteeing $\epsilon$-outage of the communication link at the receiver~\cite{goldsmith2005wireless} by using Shannon's bound.
The outage probability $\epsilon$ characterizes the probability of packet loss in case of deep fading, when the transmission cannot be decoded. 
 We assume that the data is correctly received if the instantaneous received SNR is higher than $\gamma_\text{th}$. Thus, for a threshold SNR $\gamma_\text{th}$ with outage probability $\epsilon$, the rate $R(\epsilon)$ is
\begin{equation} \label{eq:rate}
R(\epsilon) = B\log_2(1+\gamma_\text{th}).
\end{equation}
The outage probability $\epsilon$ in a Rayleigh fading channel is given by
 \begin{equation}\label{eq:epsilon}
 \epsilon = \Pr \left\{ \gamma < \gamma_\text{th} \right\} = \Pr \left[ g < \frac{\gamma_\text{th}}{\gamma_{0}} \right] = 1 - e^{\left(-\frac{\gamma_\text{th}}{\gamma_{0}}\right)}, 
\end{equation}
where $\mbox{Pr}\{ \cdot \}$ denotes the probability of an event. 
From \eqref{eq:rate} and \eqref{eq:epsilon}, $R(\epsilon)$ can be rewritten as 
\begin{equation} \label{eq:rate_epsilon} 
R(\epsilon) = B\log_2(1+\gamma_\text{th}) = B\log_2 \left( 1+\gamma_{0} \ln \left(\frac{1}{1-\epsilon} \right) \right).
\end{equation}

\begin{remark} \label{rem:rate_reliability}
$R(\epsilon)$ is a monotonically increasing function of $\epsilon$.
\end{remark} 

If the transmitter divides data into packets with a constant length $n_p$ and uses a pass-band modulation, the time $t_p$ to transmit is simply given by:
\begin{equation}
t_p = \frac{n_p}{2BR(\epsilon)}.
\end{equation} 
 As the erasure probability for a packet is $\epsilon$, the total time until correct reception is a geometrically distributed RV. The time elapsed in receiving acknowledgement is very small compared to the data transmission time, and hence it is ignored. 
Thus, the probability mass function (PMF) of the time $\mathcal{T}$ required to transmit a packet is then given by
\begin{equation}
\mbox{Pr}(\mathcal{T}=kt_p) = \epsilon^{k-1} (1-\epsilon), \;\; k \geq 1.
\end{equation} 
The mean and variance of $\mathcal{T}$ are then:
\begin{equation}
\mathbb{E}[\mathcal{T}] = \frac{1}{1-\epsilon}t_p, \;\; \mathbb{E}[(\mathcal{T}-\mathbb{E}[\mathcal{T}])^2]=\frac{\epsilon}{(1-\epsilon)^2} t_p^2.
\end{equation}
If a data block is composed of $N$ packets, the total transmission time for the block is:
\begin{equation}
T_{\text{tx}}(N, \epsilon) = \sum_{i=1}^{N} \mathcal{T}_i. 
\end{equation} 
We note that $T_{\text{tx}}$ is the sum of $N$ identical and independent geometrically distributed RVs. It's PMF is then a negative binomial distribution as follows:
\begin{equation}
 \mbox{Pr}(T_{\text{tx}}=kt_p|N,\epsilon)=\binom{k+N-1}{N-1}\epsilon^{k-N}(1-\epsilon)^N, \;\; k \geq N.
\end{equation}
In most practical cases, $N$ will be relatively large, and we can use the Central Limit Theorem to approximate this distribution to a Gaussian RV as follows:
\begin{equation} \label{eq:t_tx_dist} 
T_{\text{tx}}(N, \epsilon) \sim \mathcal{N}(\mu_{\text{tx}}, \sigma_{\text{tx}}^2), 
\end{equation} 
where $\mu_{\text{tx}} = N \frac{1}{1-\epsilon}t_p$ and $\sigma_{\text{tx}}^2 = N \frac{\epsilon}{(1-\epsilon)^2} t_p^2$.
This approximation substitutes the transmission time from a discrete domain with only positive values with the real domain.
However, as we assume $N\gg1$, the approximation error is negligible.

\begin{remark} 
The transmission time, modeled here using the Gaussian distribution, should be non-negative, but the domain of the Gaussian distribution is $(-\infty, \infty)$. 
It is known that around 99.73\% of the data points will fall within three standard deviations from the mean for a Gaussian distribution. We have noted that $\mu_{\text{tx}} - 4 \sigma_{\text{tx}} \geq 0$ for the numerical values considered in the simulation. Therefore, the left tail of the distribution has a negligible effect on the analysis, and the latency will never be negative. 

\end{remark}

\begin{remark} \label{rem:gaussian_dist_nature} 
The PDF of Gaussian distribution is neither a convex nor a concave function. It is symmetric and exhibits unimodal variation. Further, it is also a log-concave function. 
\end{remark}

\subsection{{Latency  Incurred in Data Processing }} 

An MEC server performs tasks related to computation, and specifically in this case, compression and decompression also. 
The time elapsed in these tasks depend upon the number of central processing unit (CPU) cycles required to process one bit of data, the clock frequency of the MEC server, and the volume of the data to be processed.
In the context of the considered closed-loop teleoperation system (see Fig.~\ref{fig:sys_mod_MEC_TI}), the processing capability of the MEC mounted at BS will be much higher than the robot's~\cite{robot_comput}. Therefore, the computation to extract the low-level command from the raw data is accomplished at the BS rather than at the robot. Here, the time elapsed in different processes are modelled, which will be helpful in analysing the closed-loop latency.

\subsubsection{Latency incurred in computation}
The time elapsed in computation $T_c$ to compute $D_0$ volume of data is given as 
\begin{equation} \label{eq:comp_time_1}
T_{c} = \frac{D_0 X_c}{f_{0}},
\end{equation}
where $X_c$ is the number of CPU cycles and $f_0$ is the frequency of the MEC server.

A recent study~\cite{MEC_random_var} shows that the number of cycles allocated to compute one bit at the MEC server is stochastic in nature. This is because the CPU cycles are allocated to different ongoing tasks at the MEC server simultaneously. The number of CPU cycles required to compute one bit of data is modeled in the relevant literature~\cite{MEC_gamma_1, MEC_gamma_2} as an RV following a Gamma distribution.
Thus, the PDF of $X_c$ is given as
\begin{equation} \label{eq:no_of_cycles_1}
X_c \sim \mbox{Gamma}(\kappa_1, \beta_1) \Rightarrow f_X(x; \kappa_1, \beta_1) = \frac{1}{(\beta_1)^\kappa_1 \Gamma(\kappa_1)} x^{\kappa_1-1} \exp(-{x}/{\beta_1}),
\end{equation}
where $\kappa_1$ is the shape parameter and $\beta_1$ is the scale parameter. $\Gamma(s)= \int_{0}^{\infty} t^{s-1} e^{-t} dt$ is the Gamma function.
 Note that $\mathbb{E}[X_c] = \kappa_1 \beta_1$. 

Now, from \eqref{eq:comp_time_1}, \eqref{eq:no_of_cycles_1}, and the transformation of the PDF of $X_c$, the distribution of computation time at MEC server is given as 
\begin{equation} \label{eq:t_mec_com} 
T_c \sim \mbox{Gamma}\bigg(\kappa_1, \frac{D_0 \beta_1}{f_0}\bigg) \Rightarrow f_{T_c}(t; \kappa_1, \beta_1, D_0, f_{0}) = \bigg( \frac{f_{0}}{D_0 \beta_1}\bigg)^{\kappa_1} \frac{1}{\Gamma(\kappa_1)} t^{\kappa_1-1} \exp \bigg(\frac{-tf_{0}}{D_0 \beta_1}\bigg).
\end{equation}
The expected computation delay $\bar{T}_{c}$ is obtained as follows 
\begin{equation} 
\bar{T}_{c}(\kappa_1, \beta_1, D_0, f_{0}) = \mathbb{E}[T_{c}] = \frac{D_0 \kappa_1 \beta_1}{f_{0}}.
\end{equation}

\subsubsection{Latency incurred in compression} 
The latency of data compression depends on the data volume and computational properties of the device's processor. 
 Specifically, the time elapsed $T_\text{cp}$ in compressing volume of data $D_0$ is given as \cite{jsac_comp_model} 
\begin{equation} \label{eq:comp_time}
T_\text{cp} = \frac{D_0 X_{cp}}{f_\text{0}},
\end{equation}
where $X_{cp}$ is the number of CPU cycles required to compress one bit of data, and $f_\text{0}$ is the frequency (i.e., clock speed) of the processor. 
Analogously to the previous case, $X_{cp}$ is stochastic in nature and follows the Gamma distribution given as follows
\begin{equation} \label{eq:comp_time_exp}
X_{cp} \sim \mbox{Gamma}(\kappa_2, \beta_2),
\end{equation}
where $\kappa_2$ and $\beta_2$ are respectively the shape and scale parameters. 
Note that $\mathbb{E}[X_{cp}] = \kappa_2 \beta_2$.

Thus, $T_\text{cp}$ is also an RV and its PDF is given as
\begin{equation} \label{eq:t_mec} 
\begin{split}
T_{cp} \sim \mbox{Gamma}\bigg(\kappa_1, \frac{D_0 \beta_2}{f_0}\bigg). & 
 \end{split}
\end{equation}

 A lossless data compression is assumed, such that the original data can be perfectly reconstructed from the compressed data without error\footnote{Huffman, run-length, Lempel-Ziv, and bzip2 are some of the most commonly used compression techniques to achieve lossless data compression \cite{data_comp_decomp_1}.} \cite{data_comp_decomp_4}. For lossless compression, the average number of CPU cycles required to compress one bit of raw data is given as~\cite{comm_lett_comp,jsac_comp_model} 
\begin{equation} \label{eq:compresiion_ratio}
\mathbb{E}[X_{cp}] = \kappa_2 \beta_2 = \exp(Q \psi) - \exp(\psi) = C(Q),
\end{equation} 
where $Q\geq 1$ is the compression ratio (i.e., the ratio of the sizes of raw and compressed data) and $\psi$ is a positive constant.

Using \eqref{eq:compresiion_ratio}, the PDF of compression time ${T_\text{cp}}$ is also a RV, which is given as 
\begin{equation} \label{eq:t_mec_comp} 
\begin{split}
 T_{cp} \sim \mbox{Gamma}\bigg(\kappa_2, \frac{D_0 C(Q)}{\kappa_2 f_0}\bigg) . & \\
\end{split}
\end{equation}
The expected value of ${T_\text{cp}}$, $\bar{T}_\text{cp}$ for the compression ratio $Q$ is given as 
\begin{equation} 
\bar{T}_\text{cp}(\kappa_2, \beta_2, D_0, f_{0}, Q) = \mathbb{E}[T_\text{cp}] = \frac{D_0 \mathbb{E}[X_{cp}]}{f_\text{0}} = \frac{D_0 C(Q)}{f_\text{0}}.
\end{equation}

\subsubsection{Latency incurred in decompression}
Decompression refers to the process of restoring compressed data to its original form.
It is also a type of computation, and can be performed on MEC server.
The latency incurred $T_d$ in decompressing $D_0$ amount of data is given as
\begin{equation}
T_d = \frac{D_0 X_{d}}{f_\text{0}},
\end{equation}
where $X_d$ denotes the number of cycles required to decompress one bit of data, which will also follow the Gamma distribution given as 
\begin{equation} \label{eq:decomp_time} 
X_d \sim \mbox{Gamma}(\kappa_3, \beta_3),
\end{equation}
where $\kappa_3$ and $\beta_3$ are respectively the shape and scale parameters. Note that $\mathbb{E}[X_{d}] = \kappa_3 \beta_3$.

Recent work in \cite{data_comp_decomp_1, data_comp_decomp_2, data_comp_decomp_3} shows that if same volume of data is compressed and decompressed then the time elapsed in the decompression process is less than that elapsed in the compression process. Thus, the average number of cycles required in decompression and compression process can be related as follows
\begin{equation} \label{eq:comp_decomp_1}
\mathbb{E}[X_d] = \zeta \mathbb{E}[X_{cp}], 
\end{equation}
where $0 < \zeta < 1$ is a constant.

Using \eqref{eq:comp_decomp_1}, the following can be written,
\begin{equation} \label{eq:comp_decomp}
\kappa_3 \beta_3 = \zeta \kappa_2 \beta_2 = \zeta C(Q). 
\end{equation}

Thus, the decompression time $T_d$ is also an RV, $T_\text{d} \sim \mbox{Gamma}\bigg( \kappa_3, \frac{D_0 \zeta C(Q)}{f_0 \kappa_3} \bigg)$, and its PDF with compression ratio $Q$ is given as
\begin{equation} \label{eq:t_mec_decomp} 
\begin{split}
T_\text{d} \sim \mbox{Gamma}\bigg( \kappa_3, \frac{D_0 \zeta C(Q)}{f_0 \kappa_3} \bigg) & . \\
\end{split} 
\end{equation}

The expected value of $T_d$, $\bar{T}_{d}$, for the compression ratio $Q$, $\bar{T}_{d}$, is obtained as 
\begin{equation} 
\bar{T}_{d}(\kappa_3, \beta_3, D_0, f_{0}, Q) = \mathbb{E}[T_{d}] = \frac{D_0 \kappa_3 \beta_3}{f_{0}} = \frac{D_0 \zeta \kappa_2 \beta_2}{ f_{0}} = \frac{D_0 \zeta C(Q)}{f_{0}} .
\end{equation}

\begin{remark} \label{rem:gamma_dist_nature}
The PDF of Gamma distribution is neither a convex nor a concave function, but it exhibits a unimodal variation. Also, it is not a symmetric distribution \cite{gamma_pdf}. Thus, the PDFs of $T_{c}$, $T_{cp}$, and $T_{d}$ are neither convex nor concave functions, but are all unimodal functions. 
In addition, Gamma distribution is a log-concave function, and hence the PDFs of $T_{c}$, $T_{cp}$, and $T_{d}$ are log-concave functions. 
\end{remark}

\section{Analysis of Closed-Loop Latency of MEC-enabled Teleoperation System} \label{sec:analysis}
Using the results from the previous section, the analytical framework to estimate the closed-loop latency of MEC-enabled teleoperation system shown in Fig. \ref{fig:sys_mod_MEC_TI} is developed here. We assume that the HO transmits all the command data to the BS for processing at the MEC server. On the other hand, two scenarios are analyzed regarding the processing of the raw sensing data at the robot. In the first case, the robot located in the remote environment compresses the raw sensing data first and then transmits these compressed data. The MEC server then decompresses the data to recover the original version, which is processed to extract the user readable command to be transmitted to the HO.
In the second case, the robot does not compress the sensing data, but transmits them in raw form to the BS for further processing at the MEC server to extract the user readable command for the HO.



 Let $D_c$ and $D_s$ be the volume of command and sensing signal, respectively. Let the distance between the HO and the BS be $d_{ho}$, and the same between BS and the robot be $d_{r}$.
 The transmission powers of the HO, the BS, and the robot are $P_{\text{tx}}^{ho}$, $P_{\text{tx}}^{bs}$, and $P_{\text{tx}}^{r}$, respectively. Assume that $B$ amount of bandwidth is dedicated for this closed-loop operation, and the HO, the BS, and the robot transmit over this bandwidth. In the given context, it may be noted that the compression process will be performed at the robot, whereas the decompression and computation processes will be performed at the BS.

Since all processes are executed over the same server, the shape parameter (see \eqref{eq:no_of_cycles_1}, \eqref{eq:comp_time_exp}, and \eqref{eq:decomp_time}) will remain the same for a given processor. On the other hand, the scale parameter will be different for different tasks (computation, compression, or decompression), as different amounts of resources in terms of CPU cycles need to be allocated. Let the shape parameter of the MEC associated with BS be $\kappa_{bs}$, and the scale parameter for computation and decompression at the BS be $\beta_c$ and $\beta_d$, respectively.
Further, let the shape parameter of the processor embedded with robot be $\kappa_r$, and the scale parameter for the compression process be $\beta_{cp}$. 
Finally, let the frequency of the MEC server associated with BS and processor at the robot be $f_{\text{MEC}}$ and $f_R$, respectively. 
Thus, from \eqref{eq:compresiion_ratio} and \eqref{eq:comp_decomp}, we get
\begin{equation}
\begin{split} 
& \kappa_r \beta_{cp} = C(Q), \;\;\; \kappa_{bs} \beta_{d} = \zeta C(Q).
\end{split} 
\end{equation}


\subsection{Case 1: Data Compression at Robot}
The closed-loop latency is composed of the latency incurred during transmitting the command data (from the HO to the robot) and the sensing data (from the robot to the HO). 
{The latency incurred in transmitting command data from HO to robot is composed of the data transmission time, extracting the low-level command from raw command signal at MEC associated at BS, and transmitting the low-level command to the robot from the BS.} 
Thus, referring to Fig. \ref{fig:sys_mod_MEC_TI}, the latency involved in transmitting the command signal ${T_1^c}$ when $N_c$ data packets are to be sent with outage probability $\epsilon$ is given as
\begin{equation}
T_1^c = T_{\text{tx}}^c(N_c, \epsilon) + T_c^c(\kappa_{bs}, \beta_c, D_c,f_{bs}) + T_{\text{tx}}^{pc}(N_c^p, \epsilon),
\end{equation}
where $T_{\text{tx}}^c$ is the time elapsed in transmitting the command signal, see \eqref{eq:t_tx_dist}, and $T_c^c$ is the time taken by MEC server to estimate the low-level command, see \eqref{eq:t_mec_com}.
$T_{\text{tx}}^{pc}$ denotes the time elapsed in transmitting the low-level command (consisting of $N_c^p$ packets) extracted from command signal, see \eqref{eq:t_tx_dist}.
All the constituents of $T_1^c$ are independent RVs. 


The latency incurred in transmitting sensing data from the robot to HO comprises of compression at robot, transmission of compressed data, decompressing the compressed data at MEC associated with BS, extracting the low-level command from sensing data, and transmitting the low-level command to the HO from BS. Thus, referring to Fig. \ref{fig:sys_mod_MEC_TI}, the latency $T_1^f$ incurred in transmitting the compressed sensing data when $N_f$ data packets are to be transmitted with outage $\epsilon$ is given as 
\begin{equation} 
\begin{aligned}\label{eq:clt_11}
T_1^f = {} & T_{cp}^f(\kappa_r, \beta_{cp}, D_s, f_r, Q) + T_{\text{tx}}^f(N_f, \epsilon) + T_d^f(\kappa_{bs}, \beta_d, \frac{D_s}{Q}, f_{bs}, Q) \\
& + T_c^f(\kappa_{bs}, \beta_c, D_s, f_{bs}, Q)+ T_{\text{tx}}^{pf}(N_f^p, \epsilon),
\end{aligned}
\end{equation} 
where $T_r$ denotes the time elapsed in compressing the raw sensing data with compression ratio $Q$ (see \eqref{eq:t_mec_comp}), $T_{\text{tx}}^f$ denotes the time elapsed in transmitting the compressed data (see \eqref{eq:t_tx_dist}), $T_d^f$ denotes the delay incurred in decompressing the compressed sensing data (see \eqref{eq:t_mec_decomp}), and $T_c^f$ denotes the time elapsed in processing the sensing data (see \eqref{eq:t_mec}). $T_{\text{tx}}^{pf}$ is the time elapsed in transmitting the low-level command (having number of packets $N_f^p$) extracted from sensing data to the HO (see \eqref{eq:t_tx_dist}).
All the constituents of $T_1^f$ are independent RVs. 

The volume of the low-level command is much less than the original raw data, i.e., $N_c >> N_c^p$ and $N_f >> N_f^p$, and will not contribute significantly in the closed-loop latency. Using this fact, $T_1^c$ and $T_1^f$ can be written as,
\begin{equation} 
\begin{aligned}\label{eq:delay_cd_fb}
T_1^c \approx {} & T_{\text{tx}}^c(N_c, \epsilon) + T_c^c(\kappa_{bs}, \beta_c, D_c,f_{bs}) ; \\
T_1^f \approx {} & T_{cp}^f(\kappa_r, \beta_{cp}, D_s, f_r, Q) + T_{\text{tx}}^f(N_f, \epsilon) + T_d^f(\kappa_{bs}, \beta_d, \frac{D_s}{Q}, f_{bs}, Q) + T_c^f(\kappa_{bs}, \beta_c, D_s, f_{bs}, Q) . 
\end{aligned}
\end{equation}
The closed-loop latency $T_1$ is estimated as 
\begin{equation} 
\begin{aligned}\label{eq:clt_p}
T_1 = {} & T_1^c + T_1^f.
\end{aligned}
\end{equation} 
Using \eqref{eq:delay_cd_fb}, \eqref{eq:clt_p} can be written as
\begin{equation} 
\begin{aligned}
T_1 = {} & T_{\text{tx}}^c(N_c, \epsilon) + T_c^c(\kappa_{bs}, \beta_c, D_c,f_{bs}) + T_{cp}^f(\kappa_r, \beta_{cp}, D_s, f_r, Q) + T_{\text{tx}}^f(N_f, \epsilon) \\
 {} & + T_d^f(\kappa_{bs}, \beta_d, \frac{D_s}{Q}, f_{bs}, Q) + T_c^f(\kappa_{bs}, \beta_c, D_s, f_{bs}, Q).
\end{aligned}
\end{equation} 
This can be rewritten as
\begin{equation} \label{eq:clt_alpha_0_1} 
\begin{aligned} 
T_1 = {} & \big[T_{\text{tx}}^c(N_c, \epsilon) + T_{\text{tx}}^f(N_f, \epsilon) \big] + T_{cp}^f(\kappa_r, \beta_{cp}, D_s, f_r, Q) \\
& + \left[ T_c^c(\kappa_{bs}, \beta_c, D_c,f_{bs}) + T_d^f(\kappa_{bs}, \beta_d, \frac{D_s}{Q}, f_{bs}, Q) + T_c^f(\kappa_{bs}, \beta_c, D_s, f_{bs}, Q) \right] .
\end{aligned}
\end{equation}



The expected value of $T_1$, $\mu_{T_1}$, is given as
\begin{equation} \label{eq:mean_t_1}
\begin{aligned}
\mu_{T_1} = {} & N_c \frac{1}{1-\epsilon}t_p + \frac{ D_c \kappa_{bs} \beta_c}{f_{bs}} + \frac{D_s C(Q)}{f_r} + N_f \frac{1}{1-\epsilon}t_p + \frac{\zeta D_s C(Q)}{Q f_{bs}} + \frac{ D_s \kappa_{bs} \beta_{c}}{f_{bs}} .
\end{aligned}
\end{equation}

{ It may be noted that all the latency constituents of $T_1$ (see \eqref{eq:clt_alpha_0_1}) are independent RVs, and hence the closed-loop latency $T_1$ is also an RV. $T_1$’s constituent distributions $T_c^c, T_d^f$, and $T_c^f$ follow Gamma distribution with different scale and shape parameters, whereas $T_{\text{tx}}^c$ and $T_{\text{tx}}^f$ follow Gaussian distribution. It is very difficult to estimate the closed-form expression of the PDF of the RV $T_1$. Therefore, it is very important to characterize its properties for further analysis. }

\begin{lemma} \label{lemma_cvx} 
The PDF of the sum of two independent RVs is convex iff at least one of them is convex. In the same way, the PDF of the sum of two independent RVs is concave iff at least one of them is concave. 
\end{lemma}


 \begin{proof}
 See Appendix \ref{lemma_1}.
 \end{proof}

\begin{remark} \label{rem:pdf_t_1}
The PDF of $T_1$ is neither a convex nor a concave function of $t$, because none of its constituent distributions are either convex or concave (see Lemma \ref{lemma_cvx}). 
\end{remark}

\begin{theorem}
The CDF of $T_1$ is neither a convex nor a concave function of time. 
\end{theorem}

\begin{proof}
See Appendix \ref{theorem_1}.
\end{proof}

\subsection{Case 2: Raw Data Offloading to MEC }
Here, no data compression happens at the robot, and the whole raw sensing data is transmitted to the BS, where MEC processes it in order to estimate the user readable command. Thus, referring to Fig. \ref{fig:sys_mod_MEC_TI}, the latency involved $T_2^c$ in transmitting the command signal from HO to the robot is given as
\begin{equation} 
\nonumber 
T_2^c = T_1^c.
\end{equation}

Now, referring to Fig. \ref{fig:sys_mod_MEC_TI}, the latency $T_2^f$ involved in transmitting the raw sensing data having $M_f$ packets from the robot to the HO is given as 
\begin{equation} 
\begin{aligned}\label{eq:clt_22}
T_2^f = {} & T_{\text{tx}}^f(M_f, \epsilon) + T_c^f(\kappa_{bs}, \beta_c, D_s, f_{bs})+ T_{\text{tx}}^{pf}(N_f^p, \epsilon),
\end{aligned}
\end{equation} 
where the details of the parameters are mentioned in \eqref{eq:clt_11}. 

Ignoring the latency elapsed in sending the low-level commands to robot and HO (see \eqref{eq:delay_cd_fb}), the closed-loop latency $T_2$ in this case is given as
\begin{equation} 
\begin{aligned}
T_2 = {} & T_2^c + T_2^f \\
= {} & T_{\text{tx}}^c(N_c, \epsilon) + T_c^c(\kappa_{bs}, \beta_c, D_c,f_{bs}) + T_{\text{tx}}^f(N_f, \epsilon) + T_c^f(\kappa_{bs}, \beta_c, D_s, f_{bs}). 
\end{aligned}
\end{equation} 
This can be rewritten as
\begin{equation} \label{eq:clt_alpha_1_2} 
\begin{aligned}
T_2 = {} & \big[ T_c^c(\kappa_{bs}, \beta_c, D_c,f_{bs}) + T_c^f(\kappa_{bs}, \beta_c, D_s, f_{bs}) \big] + \big[ T_{\text{tx}}^c(N_c, \epsilon) + T_{\text{tx}}^f(M_f, \epsilon) \big]. 
\end{aligned}
\end{equation}

The expected value of $T_2$, $\mu_{T_2}$, is given as
\begin{equation} \label{eq:mean_t_2}
\begin{aligned}
\begin{aligned}
\mu_{T_2} = {} & N_c \frac{1}{1-\epsilon}t_p + \frac{ D_c \kappa_{bs} \beta_c}{f_{bs}} + M_f \frac{1}{1-\epsilon}t_p + \frac{ D_s \kappa_{bs} \beta_{c}}{f_{bs}}.
\end{aligned}
\end{aligned}
\end{equation} 
As the constituents of $T_2$ are all independent RVs, we can make some of the same inferences that we proved for $T_1$.

%


\begin{remark} \label{rem:pdf_t_2}
The PDF of $T_{2}$ is neither a convex nor a concave function of time, because none of its constituent distributions are either convex or concave (see Lemma \ref{lemma_cvx}).
\end{remark}

\begin{theorem}
The CDF of $T_{2}$ is neither a convex nor a concave function of time. 
\end{theorem}

\begin{proof}
See Appendix \ref{theorem_2}.
\end{proof}

\subsection{Optimization of Closed-loop Teleoperation System}

It is important to optimize the closed-loop latency $\tau_i$ ($i=1$ for Case 1 and $i=2$ for Case 2) of the MEC-enabled teleoperation system for a given bandwidth $B$. This also requires to estimate the optimal compression ratio locally at the robot itself before transmitting it to the BS. For this purpose, an optimization problem is formulated as follows 
\begin{equation}
\nonumber 
\begin{aligned}
 \textbf{(P1)}:~~ & \min_{Q, \epsilon} \quad \tau_i, \;\;\; i \in \{ 1,2 \} \\ 
 \text{s. t.} \;\; {\textbf{(C1)}}:~ & ~ F_{T_i}(\tau_i) \geq \varrho_\text{th}, \;\;\; i \in \{ 1,2 \} \\
 {\textbf{(C2)}}:~ & ~ 0 \leq \epsilon \leq \epsilon_{th} \\
 {\textbf{(C3)}}:~ & ~ Q \geq Q_{th} =1 
\end{aligned}
\label{eq:opt_prob_p1}
\end{equation}
Constraint {\textbf{(C1)}} ensures the closed-loop latency $\tau$ in probabilistic sense, where $F_{T}(\cdot)$ denotes the CDF of $T$.
 $\varrho_\text{th}$ is an statistical parameter, which indicates the probability of the closed-loop latency be at most $\tau$. Constraint {\textbf{(C2)}} restricts the range of link outage probability. Constraint {\textbf{(C3)}} indicates the range of compression ratio. 

 In order to solve the optimization problem {\textbf{(P1)}}, we need to obtain the distribution of $T_1$ and $T_2$ to meet the statistical guarantee on the closed-loop latency criterion in {\textbf{(C1)}}. 
However, as we noted above, these PDFs are hard to obtain, because the closed-form expression of the distribution of the sum of arbitrary Gamma RVs is not known. However, several approximation methods \cite{gamma_approx_1, gamma_approx_2, gamma_approx_3} to estimate the distribution of sum of Gamma RVs are reported in the literature. These methods offer better accuracy when the number of RVs to be summed are very high. Here, the closed-loop latency $T_1$ comprises of only four Gamma RVs (see \eqref{eq:clt_alpha_0_1}), whereas the closed-loop latency $T_2$ comprises of only two Gamma RVs (see \eqref{eq:clt_alpha_1_2}). The approximation error may then be high, which is unacceptable for the teleoperation system design as we are dealing with time-critical application scenario. Therefore, these approximation methods are not viable options in the given context. However, finding lower and upper bounds on the values of $T_1$ and $T_2$ will be helpful to solve the optimization problem \textbf{(P1)}.

\begin{theorem} \label{th:bound_case_1} 
 The closed-loop latency $T_1$ for a given $\varrho_{th}$ with $F_{T_1}(\tau_1) = \varrho_{th}$ is bounded as follows:
\begin{equation}
\nonumber 
 \tau_{1,L}(Q, \epsilon, \varrho_{th}) \leq \tau_1 \leq \tau_{1,U}(Q, \epsilon, \varrho_{th})
\end{equation}
where 
\begin{equation} 
\nonumber
\begin{aligned}
 \tau_{1,L}(Q, \epsilon, \varrho_{th}) ={} & \max\left( F_{T_{tx}^c}^{-1}(\varrho_{th}), F_{T_{tx}^f}^{-1}(\varrho_{th}), F_{T_{cp}^f}^{-1}(\varrho_{th}), F_{T_{c}^c}^{-1}(\varrho_{th}), F_{T_{d}^f}^{-1}(\varrho_{th}), F_{T_{c}^f}^{-1}(\varrho_{th}) \right); \\
 \tau_{1,U}(Q, \epsilon, \varrho_{th}) ={} & \min\bigg(F_{T_{tx}^c}^{-1}(\varrho_{th}) + F_{T_{tx}^f}^{-1}(\varrho_{th}) + F_{T_{cp}^f}^{-1}(\varrho_{th}) + F_{T_{c}^c}^{-1}(\varrho_{th}) + F_{T_{d}^f}^{-1}(\varrho_{th}) + F_{T_{c}^f}^{-1}(\varrho_{th}),\\
 &\frac{\mu_{T_1}}{1-\varrho_{th}} \bigg).
\end{aligned} 
\end{equation}
 $F_{\mathcal{Z}}^{-1}(\cdot)$ denotes the inverse of CDF of RV $\mathcal{Z}$.
\end{theorem}

\begin{proof}
See Appendix \ref{theorem_4}.
\end{proof}

\begin{theorem} \label{th:bound_case_2} 
 The closed-loop latency $T_2$ for a given $\varrho_{th}$ with $F_{T_2}(\tau_2) = \varrho_{th}$ is bounded as follows:
\begin{equation}
\nonumber 
 \tau_{2,L}(\epsilon, \varrho_{th}) \leq \tau_2 \leq \tau_{2,U}(\epsilon, \varrho_{th})
\end{equation}
where 
\begin{equation}
\nonumber
\begin{split}
 \tau_{2,L}(\epsilon, \varrho_{th}) = {} & \mbox{max}\left( F_{T_{tx}^c}^{-1}(\varrho_{th}), F_{T_{tx}^f}^{-1}(\varrho_{th}), F_{T_{c}^c}^{-1}(\varrho_{th}), F_{T_{c}^f}^{-1}(\varrho_{th}) \right); \\
\tau_{2,U}(\epsilon, \varrho_{th}) = {} & \mbox{min}\left(F_{T_{tx}^c}^{-1}(\varrho_{th}) + F_{T_{tx}^f}^{-1}(\varrho_{th}) + F_{T_{c}^c}^{-1}(\varrho_{th}) + F_{T_{c}^f}^{-1}(\varrho_{th}), \frac{\mu_{T_2}}{1-\varrho_{th}} \right). 
\end{split}
\end{equation}
\end{theorem}

\begin{proof}
See Appendix \ref{theorem_5}.
\end{proof} 

The closed-loop latency is the sum of independent RVs, and it doesn’t have the closed-form expression. However, it is well known that the PDF of the sum of RVs is obtained by convolving the constituent PDFs. In continuous domain, the convolution leads to integration which is difficult to compute since it requires the solution of a complicated multidimensional integral. Therefore, the convolution in discrete time domain is adopted here due to simplicity. Towards this, the constituent PDFs are discretized by sampling with same sampling interval (here $1$ millisecond is considered) in order to obtain the approximate probability mass functions (PMF). Then, the constituent PMFs are convolved to obtain the PMF of the closed-loop latency. Thus, we choose to use the numerical technique to solve the optimization problem {\textbf{(P1)}}.
However, the estimation of the optimal value of the latency along with link outage $\epsilon$ and compression ratio $Q$ requires the exhaustive search, which is computationally extensive. Therefore, it requires a computationally-efficient procedure to solve  {\textbf{(P1)}}, which can be achieved by reducing the search space of $\epsilon$ and $Q$. 
One can observe that, as the value of $\epsilon_{\text{th}}$ increases, the time elapsed in transmitting the data decreases (see Remark \ref{rem:rate_reliability}), and hence $\epsilon = \epsilon_{\text{th}}$ will be the optimal value. Now, the problem is to find the optimal value of the compression ratio $Q$, and we will use the obtained upper and lower bounds of the closed-loop latency for this purpose. 
Using the bounds on the closed-loop latency, following can be  written
\begin{equation}
 \tau_{i,L}^{\mbox{opt}} \leq \tau_{i}^{\mbox{opt}} \leq \tau_{i,U}^{\mbox{opt}},  \;\;\; i \in \{ 1,2 \}
\end{equation}
where $\tau_{i,L}^{\mbox{opt}} = \underset{Q, \epsilon_{\text{th}}, \varrho_{th} }{\mathrm{argmin}} \; \tau_{i,L}$ and   
$\tau_{i,U}^{\mbox{opt}} = \underset{Q, \epsilon_{\text{th}}, \varrho_{th}}{\mathrm{argmin}} \; \tau_{i,U}$. 
Then, the reduced search interval of compression ratio  $Q_{i, \mbox{SI}}$   can be obtained as follows
\begin{equation}
 Q_{i, \mbox{SI}} = \big\{Q \;| \; \tau_{i,L} \leq  \tau_{i,U}^{\text{opt}} \big\}.
\end{equation} 
The $Q_{i, \mbox{SI}}$ is divided into linearly spaced equidistant points (here at an interval of $0.01$), and then the PMF of the closed-loop latency is estimated for each of the values by   convolving constituent PMFs. Finally, the optimal value of $Q$ is the one which offers the minimum closed-loop latency by satisfying the constraint {\textbf{(C1)}} of the optimization problem {\textbf{(P1)}}. 
Thus,  the knowledge of $Q_{i, \text{SI}}$  leads to reduce the search space for compression ratio, and speeds up the procedure to solve the optimization problem {\textbf{(P1)}}.

\section{Simulation Results}\label{sec:results}

We illustrate the analysis presented in the previous sections through numerical evaluations. The values of the parameters considered are: $\ell = 2$, $D_c=0.15$ Mb, $D_s=0.5$ Mb,  $f_{bs}=15$ GHz, $\kappa_{bs}=1.25$, $\kappa_{r}=1.5$, $\zeta=0.1$, $\Psi=3.5$, $B=10$ MHz, $T_0=0.5 \mu s$, $K_0=-27$ dB, $N_0=-110$ dB, $d_{r-bs} = d_{bs-ho} =2$ km, $P_{\text{tx}}^{ho}=P_{\text{tx}}^{r}=0.5$ W. The computational capabilities of the MEC-enabled BS are consistent with the NVIDIA Jetson TX1, a common embedded processor for edge computing applications~\cite{halawa2017nvidia}.

 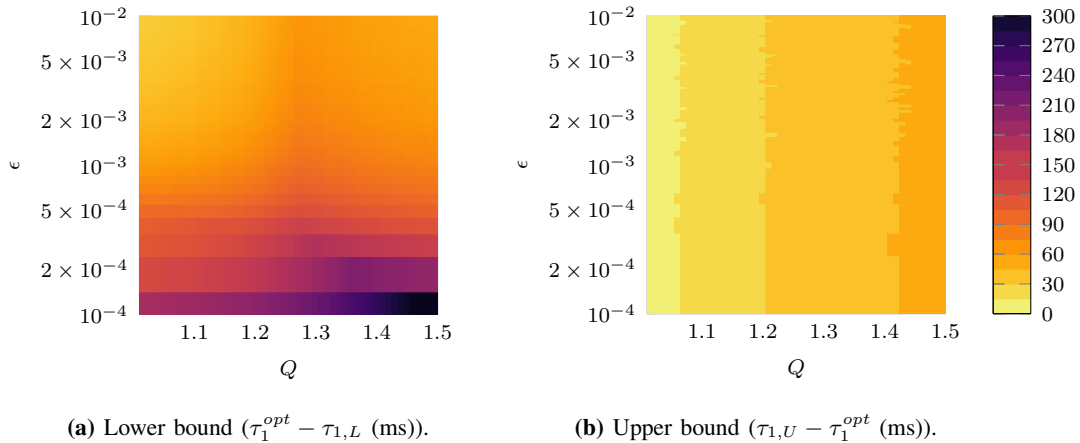
\begin{figure}[t!]
    \centering
       \begin{subfigure}[b]{.4\linewidth}
	    \flushleft
        \input{tikz_figs/case_1_lb}
        \caption{Lower bound ($\tau_1^{opt} - \tau_{1,L}$ (ms)).}
        \label{fig:case1_lb_val}
    \end{subfigure}	
	\begin{subfigure}[b]{.4\linewidth}
	    \flushright
        \input{tikz_figs/case_1_ub}
        \caption{Upper bound ($\tau_{1,U} - \tau_1^{opt}$ (ms)).}
        \label{fig:case1_ub_val}
    \end{subfigure}
\caption{ Validation of bounds on closed-loop latency for Case 1 with $f_R=1$ GHz,  $\rho_\text{th}=0.95$. }
\label{fig:bound_case_1} 
\end{figure}

 


\begin{figure}[t!]
    \centering
    \input{tikz_figs/case_2_lower_bound}
    \caption{Validation of bounds on closed-loop latency for Case 2 with   $\rho_\text{th}=0.95$.}
    \label{fig:bound_case_2}
\end{figure}
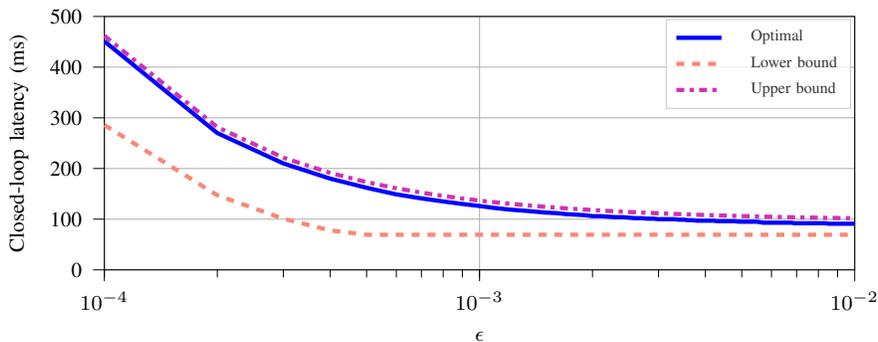

\subsection{Validation of Bounds}
The validation of the bounds on the closed-loop latency obtained for Case 1 in Theorem \ref{th:bound_case_1} is shown in Fig. \ref{fig:bound_case_1}.
Fig. \ref{fig:case1_lb_val} validates the lower bound as the difference between the optimal closed-loop latency and its lower bound, i.e.,  $\tau_1^{opt} - \tau_{1,L}$,  is always positive. On the other hand, the difference between the upper bound on the closed-loop latency and its optimal value,  i.e., $\tau_{1,U}  -  \tau_1^{opt}$, is always positive as shown in Fig.  \ref{fig:case1_ub_val} {{\footnote{ We note that same behavior is observed for any arbitrary range of $Q$ and $\epsilon$. }}}.  The bounds on the closed-loop latency obtained for Case 2 in Theorem \ref{th:bound_case_2} is shown in Fig. \ref{fig:bound_case_2}, which indicates that the  optimal value of closed-loop latency lies well within the bounds. Here, the compression ratio is $1$, because sensing data is not compressed at the robot located in the remote environment. From Fig. \ref{fig:bound_case_1} and Fig. \ref{fig:bound_case_2}, it can be noted that the optimal closed-loop latency in both cases is in the proximity of its upper bound.

\begin{figure}[t!]
    \centering
       \begin{subfigure}[b]{.45\linewidth}
	    \centering
        \input{./tikz_figs/case_1_fr_latency.tex}
        \caption{Transmission latency.}
        \label{fig:t_opt_lat}
    \end{subfigure}	
	\begin{subfigure}[b]{.45\linewidth}
	    \centering
        \input{./tikz_figs/case_1_fr_compr.tex}
        \caption{Compression latency.}
        \label{fig:t_opt_compr}
    \end{subfigure}
\caption{CDF of latency incurred in data transmission and compression.}
    \label{fig:cdf_t_tx} 
\end{figure}
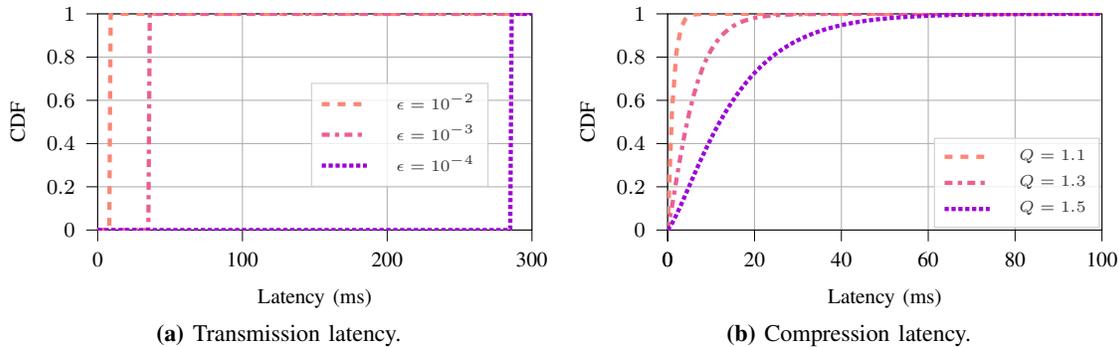

\begin{figure}[t!]
    \centering
       \begin{subfigure}[b]{.45\linewidth}
	    \centering
        \input{./tikz_figs/topt_eps.tex}
        \caption{Closed-loop latency.}
        \label{fig:t_opt_eps}
    \end{subfigure}	
	\begin{subfigure}[b]{.45\linewidth}
	    \centering
        \input{./tikz_figs/topt_q.tex}
        \caption{Optimal compression ratio.}
        \label{fig:q_opt_eps}
    \end{subfigure}
\caption{Variation of the optimal closed-loop latency and compression ratio for different cases with $\rho_\text{th}=0.95$.}
    \label{fig:t_opt_vs_reliability}
\end{figure}
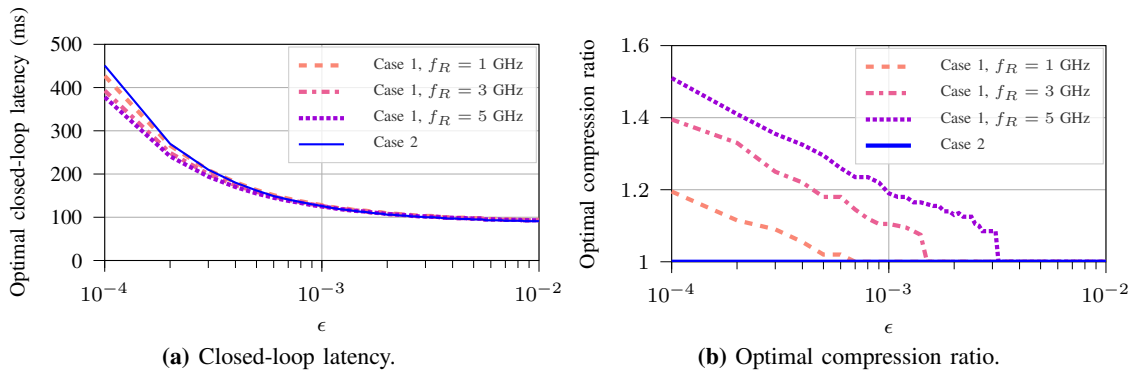

\subsection{Optimal System Design} 
The latency incurred in data transmission and compression is shown in Fig. \ref{fig:cdf_t_tx} through their CDF variation. The CDF of the latency incurred in transmitting the sensing data of volume $0.5$ Mb is shown in Fig. \ref{fig:cdf_t_tx}a for different levels of link outage. It may be noted that the transmission time increases significantly as the link outage level becomes stringent. This indicates that the data transmission with high level of accuracy demands for relatively higher transmission time. The CDF of the latency incurred in compressing the sensing data of volume $0.5$ Mb is shown in Fig. \ref{fig:t_opt_compr} for different compression ratio $Q$ with $f_R=5$GHz. The compression time increases significantly with increase in compression ratio. However, the higher compression ratio reduces the volume of sensing data which demands for relatively less time in transmission, and vice-versa. Thus, there is a trade-off between compression and transmission time.

Variation of the optimal closed-loop latency against outage $\epsilon$ is shown in Fig. \ref{fig:t_opt_eps}  for both cases.
The optimal latency is high for stringent outage requirement and it decreases as the outage probability increases. This happens because higher outage offers higher data rate which takes lower latency in data transmission.  
Although the re-transmission due to higher outage will not affect the transmission time severely due to higher data transmission rate. The effect of the computational capability of the robot has a significant impact on the latency, and the optimal value of closed-loop latency decreases as the computational capability of the robot increases.
It may be noted that the optimal latency for Case 1 is much lower than that compared to Case 2 for stringent outage requirement. 
However, the optimal latency converges towards Case 2 as the outage increases even for higher computational capability of the robot. 
This happens because higher outage offers higher data transmission rate, and hence transmission of raw sensing data is beneficial rather than compressing it. 
The optimal compression ratio against outage probability is shown in Fig.~\ref{fig:q_opt_eps}, which depends upon the computational capability of the robot as well as the outage probability. 
The optimal compression ratio decreases as outage probability increases, and the compression ratio is highest for the robot having highest computational capability.
As outage increases the optimal compression ratio converges towards $1$, i.e., no compression as in Case 2.
This observation also justifies the fact that the optimal latency converges towards Case 2 as outage increases as shown in Fig.~\ref{fig:t_opt_eps}.

\begin{remark}
Data compression is not always beneficial. The decision about whether to compress the sensing data or not depends upon the required outage constraint as well as the computational capability of the robot.
\end{remark}


\subsection{Statistical  vs  Expected Sense System Design}
The works reported in \cite{onu_1, onu_2, onu_3, coll_computing, IoT_TI, NFV_e2e_QoS} consider the system design in expected or average sense rather than in stochastic sense. Average sense design is far from the real-life deployment scenario because the data transmission over wireless network suffers from several impairments and one of them is jitter. But, the average sense design doesn’t consider the reliability criteria on random latency or jitter. 
Therefore, the reliability criteria on latency must be taken into consideration while designing the teleoperation system with low latency and high reliability requirement. Here, we will investigate the comparative analysis of the optimal closed-loop latency with different reliability criteria in stochastic sense and that in average sense for Case 1 only. However, similar shortcomings are also noted for Case 2, which is omitted due to space constraints. 
The closed-loop latency in average sense is optimized using the average latency obtained in \eqref{eq:mean_t_1} from the following 
 optimization problem  
\begin{equation}
\nonumber 
\begin{aligned}
 \textbf{(P2)}:~~ & \min_{Q, \epsilon} \quad \mu_{T_1}, \;\;  \text{s. t.} \;\; {\textbf{(C2)}} \;\; \mbox{and} \;\; {\textbf{(C3)}}  \\ 
\end{aligned}
\label{eq:opt_prob_p1_avg}
\end{equation}
The average sense design in {\textbf{(P2)}} doesn’t take into account the statistical guarantee, as the average closed-loop latency is optimized. {\textbf{(P2)}} is found to be a convex optimization problem that can be solved using CVX. The proof of convexity of {\textbf{(P2)}} is omitted for brevity.

Variation of the optimal closed-loop latency against link outage is shown in Fig. \ref{fig:comparison_with_expected_sense}a for different stochastic reliability level, such as $\varrho_\text{th}=0.95, 0.99, 0.999$. It may be noted that the optimal latency increases significantly with increase in $\varrho_\text{th}$. The amount of increase in the optimal latency is higher for $\varrho_\text{th}=0.99$ to $\varrho_\text{th}=0.95$ as compared to $\varrho_\text{th}=0.99$. The comparative view of statistical and average sense design is depicted in Fig. \ref{fig:comparison_with_expected_sense}b for $\epsilon=10^{-4}$ and $f_R=5$GHz. It may be noted that the average sense design satisfies the statistical guarantee on the closed-loop latency around $\varrho_\text{th}=0.54$ only, which is not acceptable for low-latency and high reliability applications in real-life deployment scenario. This will lead to under-resource provisioning with potential performance degradation. Thus, the system design in average sense is not suitable, {as it leads to excessively high error rates for high-reliability applications}. This justifies the consideration of statistical guarantee on low-latency and high reliability system design, as reliability of several $9$’s will be required for several applications in coming days.

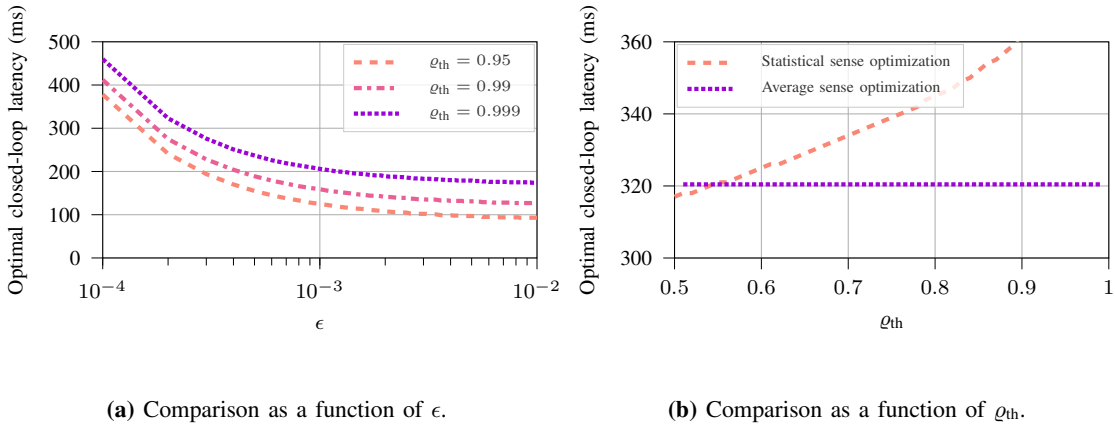
\begin{figure}[t!]
    \centering
       \begin{subfigure}[b]{0.45\linewidth}
	    \centering
        \input{./tikz_figs/topt_avg_eps.tex}
        \caption{Comparison as a function of $\epsilon$.}
        \label{fig:t_opt_vs_avg_lat}
    \end{subfigure}	
	\begin{subfigure}[b]{0.45\linewidth}
	    \centering
        \input{./tikz_figs/topt_avg_rho.tex}
        \caption{Comparison as a function of $\varrho_{\text{th}}$.}
        \label{fig:t_opt_vs_avg_rho}
    \end{subfigure}
   \caption{Performance comparison of the proposed framework with average sense design for Case 1 with $f_R = 5$ GHz.} 
    \label{fig:comparison_with_expected_sense}
\end{figure}

\begin{remark}
The system design in average sense leads to under-provisioning and a potential performance degradation, which may have severe consequences in ultra-low latency and high reliability applications.
\end{remark}



\section{Concluding Remarks}\label{sec:conc}
In this paper we have introduced a framework to analyze the closed-loop latency of the teleoperation system, where the command data from the HO and sensing data from the robot are exchanged over a wireless connection through a BS with MEC capabilities.
The sensing data is compressed before transmitting it to the BS where it decompressed first to get the original raw data.
The high-level command from the HO and the sensing data from the robot are processed by the BS to extract the low-level command for the robot and the feedback signal, respectively, and then sent to the robot, and HO, respectively. 
We have analyzed the closed-loop latency, which is found to be a sum of several RVs. Thus, the closed-loop latency is also an RV, which demands for statistical guarantee (termed as reliability) on the closed-loop latency. 
We obtained upper and lower bounds to its distribution, and formulated an optimization problem to control the transmission rate and compression ratio.
We then investigated different trade-offs in the achievable performance in terms of latency, link outage, and transmission reliability: the decision on whether to compress mostly depends on the computational capability of the robot and link outage probability.  We have also observed the shortcomings of the design approaches that only consider the expected value of the latency, disregarding worst-case outcomes.

Further work includes investigations of the age of the closed-loop of the teleoperation system from the information freshness perspective. Another interesting direction is the consideration of heterogeneous sensing data having different compression profiles. 


\ifCLASSOPTIONcaptionsoff
 \newpage
\fi

\appendix

\subsection{Proof of Lemma 1}
\label{lemma_1}
Let $\mathcal{W}$ be the distribution of sum of two independent RVs $\mathcal{U}$ and $\mathcal{V}$, so that $f_\mathcal{W}(t) = \int_{-\infty}^{\infty} f_\mathcal{V}(\tau)f_\mathcal{U}(t-\tau) d \tau$. Now, let $f_\mathcal{U}(t)$ be a convex function of $t$. We then have $f_\mathcal{U}((1-\theta) t_1 + \theta t_2) \leq (1-\theta) f_\mathcal{U}(t_1) + \theta f_\mathcal{U}(t_2), 0 \leq \theta \leq 1$. Now, $f_\mathcal{W}(1-\theta) t_1 + \theta t_2)$ is given as 
\begin{equation} \nonumber
\begin{aligned}
f_\mathcal{W}(t)((1-\theta) t_1 + \theta t_2) = {} & \int_{-\infty}^{\infty} f_\mathcal{V}(\tau)f_\mathcal{U}((1-\theta) t_1 + \theta t_2 - \tau) d \tau . \\
 = {} & \int_{-\infty}^{\infty} f_\mathcal{V}(\tau)f_\mathcal{U}((1-\theta)(t_1-\tau) + \theta (t_2-\tau)) d \tau .
 \end{aligned}
\end{equation}
Using the convexity of $f_\mathcal{U}(t)$, $f_\mathcal{W}((1-\theta) t_1 + \theta t_2)$ can be written as:
\begin{equation} \nonumber
\begin{aligned}
f_\mathcal{W}((1-\theta) t_1 + \theta t_2) \leq {} & (1-\theta) \int_{-\infty}^{\infty} f_\mathcal{V}(\tau) f_\mathcal{U}(t_1-\tau) d \tau + \theta \int_{-\infty}^{\infty} f_\mathcal{V}(\tau) f_\mathcal{U}(t_2-\tau) d \tau \\
 = {} & (1-\theta) f_\mathcal{W}(t_1) + \theta f_\mathcal{W}(t_2) .
 \end{aligned}
\end{equation}
Hence, $f_\mathcal{W}(t)$ is a convex function of $t$. 
 The concave nature can be proven in the same way.

\subsection{Proof of Theorem 1}
\label{theorem_1} 

Let $F_{T_1}(t)$ denote the CDF of $T_1$, which is given as: $F_{T_1}(t) = \int_{-\infty}^{t} f_{T_1}(x)dx$, where $f_{T_1}(\cdot)$ denotes the PDF of $T_1$. The second derivative of $F_{T_1}(t)$ is $\frac{d^2}{dt^2} F_{T_1}(t) = \frac{d}{dt} f_{T_1}(t) = f_{T_1}'(t)$.
For a function to be convex (concave), its second derivative should be non-negative (non-positive). Thus, the convexity or concavity of $F_{T_1}(t)$ depends upon the nature of $f_{T_1}'(t)$. If $f_{T_1}'(t)$ would be convex (concave) then its first derivative will be positive if and only it would be bending upward (downward), whereas its first derivative will be negative if it would be bending upward (downward).
But $f_{T_1}(t)$ is neither a convex nor a concave function of $t$ (see Remark \ref{rem:pdf_t_1}), which does not indicate about the variation of $f_{T_1}'(t)$.

\begin{table} [!b] 
\begin{center}
\caption{ Nature of different constituent distributions of closed-loop latency $T_1$ (see \eqref{eq:clt_alpha_0_1}) }
\label{tab:nature_rv} 
\begin{tabular}{p{0.25\textwidth}|p{0.12\textwidth}p{0.12\textwidth}p{0.12\textwidth}}
\toprule
 & \textbf{Unimodal} & \textbf{Symmetric} & \textbf{Log-concave} \\
\midrule
$T_{\text{tx}}^c, T_{\text{tx}}^f$ (cf. Remark \ref{rem:gaussian_dist_nature}) & Yes & Yes & Yes \\
$T_c^c$ (cf. Remark \ref{rem:gamma_dist_nature}) & Yes & No & Yes \\
$T_d^f$ (cf. Remark \ref{rem:gamma_dist_nature}) & Yes & No & Yes \\
$T_c^f$ (cf. Remark \ref{rem:gamma_dist_nature}) & Yes & No & Yes \\ 
$T_{cp}^f$ (cf. Remark \ref{rem:gamma_dist_nature}) & Yes & No & Yes \\ 
\bottomrule
\end{tabular} 
\end{center} 
\end{table}

The RV $T_1$ given in \eqref{eq:clt_alpha_0_1} is the sum of six RVs. The properties of the PDFs of these six constituent RVs are listed in Table \ref{tab:nature_rv}. It is known that the PDF of sum of RVs is the convolution of their PDFs. Therefore, the nature of the constituent PDF will determine the characteristics of the resulting PDF. 
 The study in \cite{log_concave_1, unimodal_convex} investigated the nature of the convolution of two functions with following key observations: firstly, the convolution of two log-concave function is a log-concave function, secondly, the convolution of a log-concave function and a unimodal function is also a unimodal function, and thirdly, the convolution of two asymmetric unimodal functions is a multi-modal function. 
Based on these findings and the nature of constituent PDFs of $T_1$ listed in Table \ref{tab:nature_rv}, one can deduce that $f_{T_1}(t)$ is a multi-modal function of $t$. Thus, the first derivative of $f_{T_1}(t) = f_{T_1}'(t)$ changes its sign multiple times in the domain of definition. Hence, $F_{T_1}(t)$ is neither a convex nor a concave function of $t$.

\subsection{Proof of Theorem 2} \label{theorem_2} 
 This can also be proven in the same way as Theorem 1, following the steps in Appendix \ref{theorem_1}.

\subsection{Proof of Theorem 3}
\label{theorem_4} 

Let $\mathcal{X}$ and $\mathcal{Y}$ be two non-negative independent RVs with their respective PDFs/CDFs $f_\mathcal{X}/F_\mathcal{X}$ and $f_\mathcal{Y}/F_\mathcal{Y}$. Let $\mathcal{Z} = \mathcal{X} + \mathcal{Y}$ with PDF/CDF $f_\mathcal{Z}/F_\mathcal{Z}$, which are calculated as follows:
 \begin{equation}
 \nonumber
\begin{split}
f_\mathcal{Z}(z) = {} & \int_{-\infty}^{\infty} f_\mathcal{X}(x) f_\mathcal{Y}(z-x) dx, \; F_\mathcal{Z}(z) = \int_{-\infty}^{\infty} f_\mathcal{X}(x) F_\mathcal{Y}(z-x) dx,
\end{split}
\end{equation}
Let $F_\mathcal{X}^{-1}(\varrho_{th}) = x_0 \geq 0$, $F_\mathcal{Y}^{-1} (\varrho_{th})=y_0 \geq 0$ and $F_\mathcal{Z}^{-1} (\varrho_{th})=z_0 \geq 0$. Now, $F_\mathcal{Z}(z=x_0+y_0)$ can be calculated as 
\begin{equation}
 \nonumber 
\begin{aligned}
F_\mathcal{Z}(z=x_0+y_0) = {} & \int_{-\infty}^{\infty} f_\mathcal{X}(x) F_\mathcal{Y}(x_0+y_0-x) dx = \int_{-\infty}^{\infty} f_\mathcal{X}(x_0-w) F_\mathcal{Y}(y_0+w) dw. \\ 
\end{aligned}
\end{equation}
It is known that $\mathcal{X}$ and $\mathcal{Y}$ are non-negative RVs, i.e, $f_\mathcal{X}(x)=0 \; \mbox{for} \; x<0$ and $f_\mathcal{Y}(y)=0 \; \mbox{for} \; y<0$. Thus, $F_\mathcal{Z}(z=x_0+y_0)$ can be written as:  $F_\mathcal{Z}(z=x_0+y_0) =  \int_{-y_0}^{x_0 } f_\mathcal{X}(x_0-w) F_\mathcal{Y}(y_0+w) dw$.
As $0 \leq F_\mathcal{Y}(y_0+w) \leq 1$, using the mean value theorem, $F_\mathcal{Z}(z=x_0+y_0)$ can be rewritten as
\begin{equation} 
\nonumber 
\begin{aligned}
F_\mathcal{Z}(z=x_0+y_0) = {} & K_1 \int_{-y_0}^{x_0 } f_\mathcal{X}(x_0-w) dw = K_1 F_\mathcal{X}(x_0+y_0), \; 0 \leq K_1 \leq 1. \\
\end{aligned}
\end{equation}
Based on the above finding, we have the following facts: firstly, $F_\mathcal{Z}(x_0+y_0) \leq F_\mathcal{X}(x_0+y_0)$, because $0 \leq K_1 \leq 1$, and secondly, $F_\mathcal{X}(x_0+y_0) \geq \varrho_{th} = F_\mathcal{Z}(z_0)$, because $F_\mathcal{X}(x_0)=\varrho_{th}$ and $x_0, y_0 \geq 0$. 
 Based on these facts, the following can be written 
\begin{equation}\label{eq:bound_z}
F_\mathcal{Z}(x_0+y_0) \leq F_\mathcal{Z}(z_0) \Rightarrow z_0 \leq x_0 + y_0 .
\end{equation}
Now, $F_\mathcal{Z}(z=x_0)$ can be obtained as: $F_\mathcal{Z}(z=x_0) =   \int_{-\infty}^{\infty} f_\mathcal{X}(x) F_\mathcal{Y}(x_0-x) dx$.   Using the mean value theorem, $F_\mathcal{Z}(z=x_0)$ can be rewritten as
\begin{equation}
\nonumber
F_\mathcal{Z}(z=x_0) = K_2 \int_{0}^{x_0} f_\mathcal{X}(x) dx = K_2 F_\mathcal{X}(x_0), \;\; 0 \leq K_2 \leq 1 .
\end{equation}
This implies that
\begin{equation}\label{eq:bound_x}
F_\mathcal{Z}(x_0) \leq F_\mathcal{X}(x_0) = \varrho_{th} = F_\mathcal{Z}(z_0) \Rightarrow z_0 \geq x_0 .
\end{equation}
In the same way, $F_\mathcal{Z}(y_0)$ can also be obtained:
\begin{equation} \label{eq:bound_y}
F_\mathcal{Z}(y_0) \leq F_\mathcal{Y}(x_0 + y_0) = \varrho_{th} = F_\mathcal{Z}(z_0) \Rightarrow z_0 \geq y_0.
\end{equation}
Thus, using \eqref{eq:bound_z}, \eqref{eq:bound_x}, and \eqref{eq:bound_y}, we can finally write
\begin{equation} \label{eq:bound_x_y_z}
\mbox{max}(x_0, y_0) \leq w_0 \leq x_0 + y_0 .
\end{equation}
Following the above mentioned procedure, if $\mathcal{W} = \mathcal{X} + \mathcal{Y} + \mathcal{S}$ is an RV with $F_\mathcal{W}^{-1}(w_0) = \varrho_{th}$ and $F_\mathcal{S}^{-1}(s_0) = \varrho_{th}$, then the following can be written:
\begin{equation} \label{eq:sum_of_three}
\nonumber 
\mathcal{W} = (\mathcal{X} + \mathcal{Y}) + \mathcal{S} = \mathcal{Z} + \mathcal{S}.
\end{equation}
From \eqref{eq:bound_x_y_z} and \eqref{eq:sum_of_three}, we can obtain the following bound
\begin{equation}\label{eq:bound_x_y_z_w}
\mbox{max}(z_0, s_0) \leq w_0 \leq z_0 + s_0 \Rightarrow \mbox{max}(x_0, y_0, s_0) \leq w_0 \leq x_0 + y_0 + s_0.
\end{equation} 
Thus, in the same way, the bound in \eqref{eq:bound_x_y_z_w} can be obtained for the sum of large numbers of independent non-negative RVs.

Now, if $\mathcal{W}$ be a non-negative RV having mean $\mu_\mathcal{W}$ then using Markov's inequality, the upper bound on  $\mathcal{W}$ can be given as: $\mbox{Pr}\{\mathcal{W} \geq w \} \leq \frac{\mu_\mathcal{W}}{w} \Rightarrow 1 - F_{\mathcal{W}}(w) \leq \frac{\mu_\mathcal{W}}{w} \Rightarrow F_{\mathcal{W}}(w) \geq 1 - \frac{\mu_\mathcal{W}}{w}.$   If $F_{\mathcal{W}}(w) \geq \varrho_{th}$ (which is the case in constraint {\textbf{(C1)}} of optimization problem {\textbf{(P1)}}), following inequality can be written
\begin{equation} \label{eq:markov_inequality}
1 - \frac{\mu_\mathcal{W} }{w} \leq \varrho_{th} \Rightarrow w \leq \frac{\mu_\mathcal{W}}{1-\varrho_{th}} 
\end{equation}

From \eqref{eq:bound_x_y_z_w} and \eqref{eq:markov_inequality}, following can be written 
\begin{equation} \label{eq:bound_on_three}
\mbox{max}(x_0, y_0, s_0) \leq w_0 \leq \mbox{min}(x_0 + y_0 + s_0, \frac{\mu_\mathcal{W}}{1-\varrho_{th}}) 
\end{equation}

Finally, the closed-loop latency $T_1$ (see \eqref{eq:clt_alpha_0_1}) is the sum of six non-negative independent RVs. Using the findings 
\eqref{eq:bound_x_y_z_w} ,\eqref{eq:markov_inequality}, and \eqref{eq:bound_on_three}, the bounds for RV $T_1$ can be obtained.

\subsection{Proof of Theorem 4}
\label{theorem_5} 
 This can also be proven in the same way as Theorem 3, following the steps in Appendix \ref{theorem_4}.

\bibliographystyle{IEEEtran}
\bibliography{ref_final}

\end{document}

%% file: tikz_figs/sysmodel.tex
\begin{tikzpicture}

\node (usr) at (0,0) {\includegraphics[width=1cm]{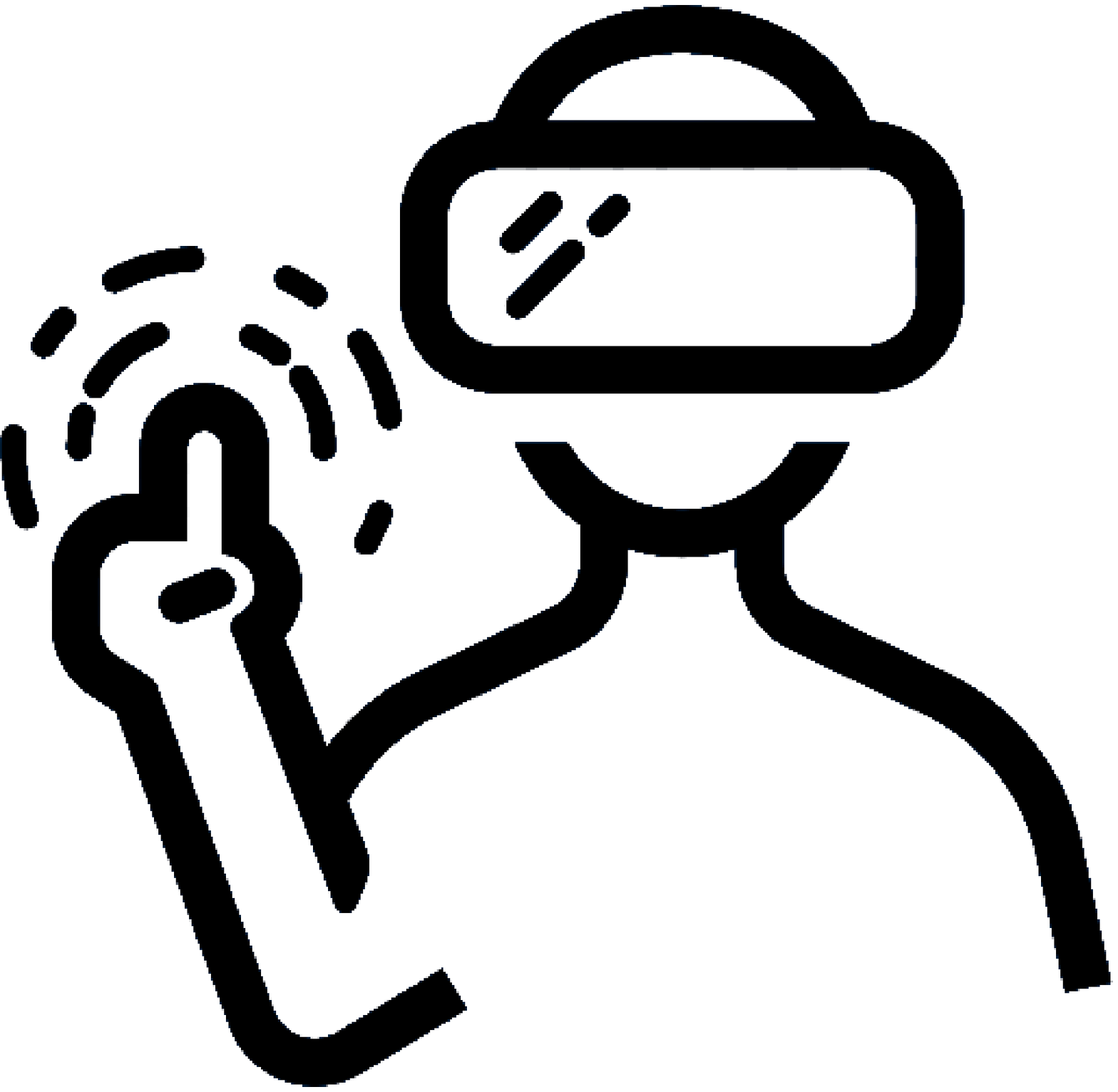}};
\node (bs) at (5,2) {\includegraphics[width=2.5cm]{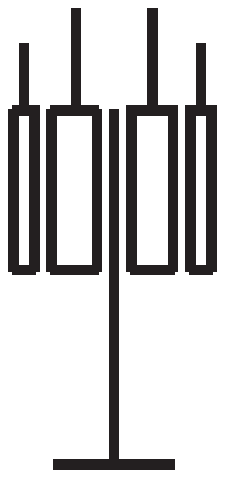}};
\node (srv) at (5,0.25) {\includegraphics[width=1cm]{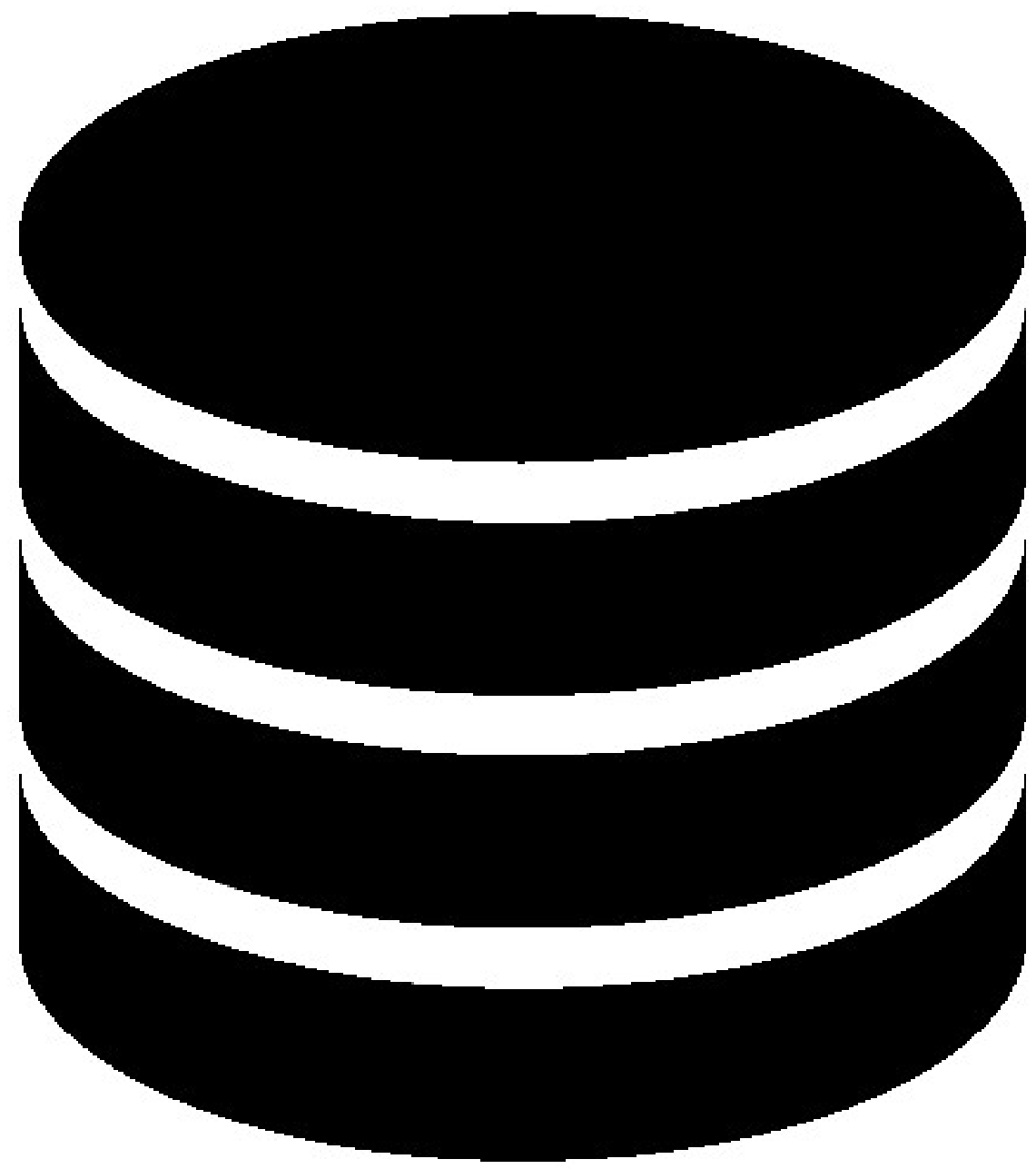}};
\node (robot) at (11,0.5) {\includegraphics[width=2.5cm]{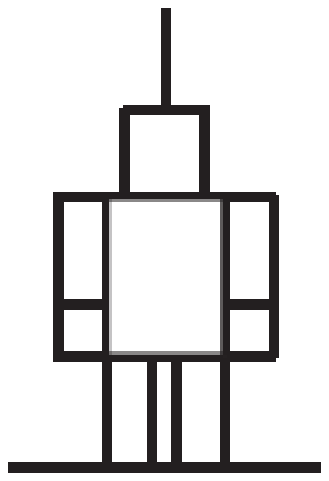}};
\node (rob) at (10,0) {\includegraphics[width=1cm]{tikz_figs/cliparts/srv.eps}};

\node[draw] (bsl) at (5,-0.5) {\scriptsize MEC-enabled BS};
\node[draw] (usl) at (0,-1) {\scriptsize HO};
\node[draw] (rol) at (10.5,-1) {\scriptsize Robot};
\node (comp) at (10,1) {\scriptsize Compression};

\draw[->] (usr.east) to node[near end,left,font={\scriptsize}] {High-level command} ([xshift=0.5cm]bs.west);
\draw[->] ([xshift=-0.5cm]bs.east) to node[near start,right,font={\scriptsize}] {Low-level command} (rob.west);
\draw[<-, dashed] ([xshift=-0.5cm,yshift=-0.2cm]bs.east) to node[near end,left,font={\scriptsize}] {Sensing data} ([yshift=-0.2cm]rob.west);
\draw[<-,dashed] ([yshift=-0.2cm]usr.east) to node[near start,right,font={\scriptsize}] {Processed feedback} ([xshift=0.5cm,yshift=-0.2cm]bs.west);
\draw[->,dotted](comp.south) to (rob.north);

\end{tikzpicture}

%% file: tikz_figs/latency_components.tex
\begin{tikzpicture}[auto]

\node[name=hmd_t] at (0,5) {HO};
\node[name=hmd_b] at (0,-3.5) {};
\node[name=bs_t] at (4,5) {MEC BS};
\node[name=cl_t] at (8,5) {Robot};
\node[name=bs_b] at (4,-3.5) {};
\node[name=cl_b] at (8,-3.5) {};

\draw[-] (hmd_t) to (hmd_b);
\draw[-] (bs_t) to (bs_b);
\draw[-] (cl_t) to (cl_b);

\draw[dashed,->,red, line width= 1.25pt] (0,4.5) to (4,3.5);
 \draw[|-|] (0,4.5) to node[midway,left] {\tiny Transmission} (0,3.5);
 \draw[|-|] (4,3.5) to node[midway,left] {\tiny MEC processing} (4,3);
\draw[dashed,->,red, line width= 1.25pt] (4,3) to (8,2.5);
 \draw[-|] (4,3) to node[midway,left] {\tiny Transmission} (4,2.5);
 \draw[|-|] (8,2.5) to node[midway,left] {\tiny Command execution} (8,2.2);
 \draw[-|] (8,2.2) to node[midway,left] {\tiny Compression} (8,1);
 
\draw[dashed,->,red, line width= 1.25pt] (8,1) to (4,0);
 \draw[-|] (8,1) to node[midway,left] {\tiny Transmission} (8,0);
 \draw[|-|] (4,0) to node[midway,left] {\tiny Decompression} (4,-0.5);
 \draw[-|] (4,-0.5) to node[midway,left] {\tiny MEC processing} (4,-1.5);

\draw[dashed,->,red, line width= 1.25pt] (4,-1.5) to (0,-2);
 \draw[-|] (4,-1.5) to node[midway,left] {\tiny Transmission} (4,-2);
 
  \draw[|-|] (0,-2) to node[midway,left] {\tiny Feedback display} (0,-2.3);

\draw[->] (-2,5) to (-2,-3.5);

\node at (-1.5,-3) {\scriptsize Time};

 \draw[|-|] (8.5,4.5) to node[midway,right] {\scriptsize Command} (8.5,2.2);
 
 \draw[|-|] (10,4.5) to node[midway,right] {\scriptsize Closed-loop latency} (10,-2.3);

 \draw[densely dotted] (0,4.5) to (8.25,4.5);
 \draw[densely dotted] (0,3.5) to (8.25,3.5);
 \draw[densely dotted] (4,3) to (8.25,3);
 \draw[densely dotted] (8.25,2.5) to (4,2.5);
 \draw[densely dotted] (8,2.2) to (8.25,2.2);
 \draw[densely dotted] (8,1) to (8.25,1);
 \draw[densely dotted] (4,0) to (8.25,0);
 \draw[densely dotted] (4,-0.5) to (8.25,-0.5);
 \draw[densely dotted] (4,-1.5) to (8.25,-1.5);
 \draw[densely dotted] (0,-2) to (8.25,-2);
 \draw[densely dotted] (0,-2.3) to (8.25,-2.3);
 
  \draw[-|] (8.5,2.2) to node[midway,right] {\scriptsize Feedback} (8.5,-2.3);

%

%

\end{tikzpicture}

%% file: tikz_figs/case_1_lb.tex
\begin{tikzpicture} 
    \begin{axis}[
    width=\sfheight,
    height=\sfheight,
    name=lin,
    scale only axis,
    xlabel=$Q$,
    ylabel=$\epsilon$,
    mesh/cols=100,
    tick pos=left,
    mesh/rows=50,
    xmin=1.01,
    xmax=1.5,
    ymin=-4,
    ymax=-2,
    xtick={1.1,1.2,1.3,1.4,1.5},
    ytick={-4,-3.7,-3.3,-3,-2.7,-2.3,-2},
    yticklabels={$10^{-4}$,$2\times10^{-4}$,$5\times10^{-4}$,$10^{-3}$,$2\times10^{-3}$,$5\times10^{-3}$,$10^{-2}$},
    point meta min=0,
    point meta max=0.3,
colormap={mymap}{[1pt]
 rgb(0pt)=(0.946403,0.937159,0.458592);
  rgb(1pt)=(0.962517,0.851476,0.285546);
  rgb(2pt)=(0.981173,0.759135,0.156863);
  rgb(3pt)=(0.987945,0.667748,0.058329);
  rgb(4pt)=(0.981895,0.579392,0.02625);
  rgb(5pt)=(0.961293,0.488716,0.084289);
  rgb(6pt)=(0.929644,0.411479,0.145367);
  rgb(7pt)=(0.886302,0.342586,0.202968);
  rgb(8pt)=(0.832299,0.283913,0.257383);
  rgb(9pt)=(0.769556,0.236077,0.307485);
  rgb(10pt)=(0.694627,0.195021,0.354388);
  rgb(11pt)=(0.621685,0.164184,0.388781);
  rgb(12pt)=(0.547157,0.136929,0.413511);
  rgb(13pt)=(0.472328,0.110547,0.428334);
  rgb(14pt)=(0.397674,0.083257,0.433183);
  rgb(15pt)=(0.316282,0.05349,0.425116);
  rgb(16pt)=(0.238273,0.036621,0.396353);
  rgb(17pt)=(0.15585,0.044559,0.325338);
  rgb(18pt)=(0.081962,0.043328,0.215289);
  rgb(19pt)=(0.025793,0.019331,0.10593);
}
]
    \addplot[matrix plot*, point meta=explicit] file [meta=index 2] {./tikz_figs/lower_bound_dist_log.csv};
\end{axis}

\end{tikzpicture}

%% file: tikz_figs/case_1_ub.tex
\begin{tikzpicture} 
    \begin{axis}[
    width=\sfheight,
    height=\sfheight,
    scale only axis,
    name=lin,
    xlabel=$Q$,
    ylabel=$\epsilon$,
    mesh/cols=100,
    tick pos=left,
    mesh/rows=50,
    xmin=1.01,
    xmax=1.5,
    ymin=-4,
    ymax=-2,
    xtick={1.1,1.2,1.3,1.4,1.5},
    ytick={-4,-3.7,-3.3,-3,-2.7,-2.3,-2},
    yticklabels={$10^{-4}$,$2\times10^{-4}$,$5\times10^{-4}$,$10^{-3}$,$2\times10^{-3}$,$5\times10^{-3}$,$10^{-2}$},
    point meta min=0,
    point meta max=300,
colormap={mymap}{[1pt]
 rgb(0pt)=(0.946403,0.937159,0.458592);
  rgb(1pt)=(0.962517,0.851476,0.285546);
  rgb(2pt)=(0.981173,0.759135,0.156863);
  rgb(3pt)=(0.987945,0.667748,0.058329);
  rgb(4pt)=(0.981895,0.579392,0.02625);
  rgb(5pt)=(0.961293,0.488716,0.084289);
  rgb(6pt)=(0.929644,0.411479,0.145367);
  rgb(7pt)=(0.886302,0.342586,0.202968);
  rgb(8pt)=(0.832299,0.283913,0.257383);
  rgb(9pt)=(0.769556,0.236077,0.307485);
  rgb(10pt)=(0.694627,0.195021,0.354388);
  rgb(11pt)=(0.621685,0.164184,0.388781);
  rgb(12pt)=(0.547157,0.136929,0.413511);
  rgb(13pt)=(0.472328,0.110547,0.428334);
  rgb(14pt)=(0.397674,0.083257,0.433183);
  rgb(15pt)=(0.316282,0.05349,0.425116);
  rgb(16pt)=(0.238273,0.036621,0.396353);
  rgb(17pt)=(0.15585,0.044559,0.325338);
  rgb(18pt)=(0.081962,0.043328,0.215289);
  rgb(19pt)=(0.025793,0.019331,0.10593);
 },
    colorbar sampled,
    colormap access=piecewise constant,
    colorbar right,
    colorbar style={samples=21,
                    ylabel={},
                    ytick={0,30,...,300}}
]
    \addplot[matrix plot*, point meta=explicit] file [meta=index 2] {./tikz_figs/upper_bound_dist_log.csv};
\end{axis}

\end{tikzpicture}
\vspace{-0.725cm}

%% file: tikz_figs/case_2_lower_bound.tex
\begin{tikzpicture}

\begin{axis}[
width=\fwidth,
height=\fheight,
xmode=log,
legend cell align={left},
legend style={at={(0.995,0.99)}, anchor=north east, fill opacity=0.8, legend columns=1, draw opacity=1, text opacity=1, font=\tiny, draw=white!80!black},
tick align=outside,
tick pos=left,
x grid style={white!69.0196078431373!black},
ylabel={Closed-loop latency (ms)},
ymin=0, ymax=0.5,
ytick={0,0.1,0.2,0.3,0.4,0.5},
yticklabels={0,100,200,300,400,500},
xtick style={color=black},
y grid style={white!69.0196078431373!black},
xlabel={$\epsilon$},
xmin=1e-4, xmax=1e-2,
ytick style={color=black},
xmajorgrids,
ymajorgrids
]

\addplot[ultra thick, color6]
table{0.0001	0.451
0.0002	0.27
0.0003	0.21
0.0004	0.18
0.0005	0.162
0.0006	0.149
0.0007	0.141
0.0008	0.135
0.0009	0.13
0.001	0.126
0.0011	0.122
0.0012	0.119
0.0013	0.117
0.0014	0.115
0.0015	0.113
0.0016	0.112
0.0017	0.11
0.0018	0.109
0.0019	0.108
0.002	0.106
0.0021	0.106
0.0022	0.105
0.0023	0.104
0.0024	0.104
0.0025	0.103
0.0026	0.102
0.0027	0.102
0.0028	0.101
0.0029	0.101
0.003	0.1
0.0031	0.1
0.0032	0.1
0.0033	0.1
0.0034	0.099
0.0035	0.099
0.0036	0.098
0.0037	0.098
0.0038	0.097
0.0039	0.097
0.004	0.097
0.0041	0.097
0.0042	0.097
0.0043	0.096
0.0044	0.096
0.0045	0.096
0.0046	0.096
0.0047	0.096
0.0048	0.096
0.0049	0.095
0.005	0.095
0.0051	0.095
0.0052	0.095
0.0053	0.095
0.0054	0.095
0.0055	0.095
0.0056	0.094
0.0057	0.093
0.0058	0.093
0.0059	0.093
0.006	0.093
0.0061	0.093
0.0062	0.093
0.0063	0.093
0.0064	0.093
0.0065	0.093
0.0066	0.093
0.0067	0.093
0.0068	0.093
0.0069	0.092
0.007	0.092
0.0071	0.092
0.0072	0.092
0.0073	0.092
0.0074	0.092
0.0075	0.092
0.0076	0.092
0.0077	0.092
0.0078	0.092
0.0079	0.092
0.008	0.092
0.0081	0.092
0.0082	0.092
0.0083	0.092
0.0084	0.092
0.0085	0.091
0.0086	0.091
0.0087	0.091
0.0088	0.091
0.0089	0.091
0.009	0.091
0.0091	0.091
0.0092	0.091
0.0093	0.091
0.0094	0.091
0.0095	0.091
0.0096	0.091
0.0097	0.091
0.0098	0.091
0.0099	0.091
0.01	0.091
};
\addlegendentry{Optimal}

\addplot [ultra thick, color2,dashed]
table {%
0.0001	0.28592
0.0002	0.14724
0.0003	0.10096
0.0004	0.077778
0.0005	0.069281
0.0006	0.069281
0.0007	0.069281
0.0008	0.069281
0.0009	0.069281
0.001	0.069281
0.0011	0.069281
0.0012	0.069281
0.0013	0.069281
0.0014	0.069281
0.0015	0.069281
0.0016	0.069281
0.0017	0.069281
0.0018	0.069281
0.0019	0.069281
0.002	0.069281
0.0021	0.069281
0.0022	0.069281
0.0023	0.069281
0.0024	0.069281
0.0025	0.069281
0.0026	0.069281
0.0027	0.069281
0.0028	0.069281
0.0029	0.069281
0.003	0.069281
0.0031	0.069281
0.0032	0.069281
0.0033	0.069281
0.0034	0.069281
0.0035	0.069281
0.0036	0.069281
0.0037	0.069281
0.0038	0.069281
0.0039	0.069281
0.004	0.069281
0.0041	0.069281
0.0042	0.069281
0.0043	0.069281
0.0044	0.069281
0.0045	0.069281
0.0046	0.069281
0.0047	0.069281
0.0048	0.069281
0.0049	0.069281
0.005	0.069281
0.0051	0.069281
0.0052	0.069281
0.0053	0.069281
0.0054	0.069281
0.0055	0.069281
0.0056	0.069281
0.0057	0.069281
0.0058	0.069281
0.0059	0.069281
0.006	0.069281
0.0061	0.069281
0.0062	0.069281
0.0063	0.069281
0.0064	0.069281
0.0065	0.069281
0.0066	0.069281
0.0067	0.069281
0.0068	0.069281
0.0069	0.069281
0.007	0.069281
0.0071	0.069281
0.0072	0.069281
0.0073	0.069281
0.0074	0.069281
0.0075	0.069281
0.0076	0.069281
0.0077	0.069281
0.0078	0.069281
0.0079	0.069281
0.008	0.069281
0.0081	0.069281
0.0082	0.069281
0.0083	0.069281
0.0084	0.069281
0.0085	0.069281
0.0086	0.069281
0.0087	0.069281
0.0088	0.069281
0.0089	0.069281
0.009	0.069281
0.0091	0.069281
0.0092	0.069281
0.0093	0.069281
0.0094	0.069281
0.0095	0.069281
0.0096	0.069281
0.0097	0.069281
0.0098	0.069281
0.0099	0.069281
0.01	0.069281
};
\addlegendentry{Lower bound}

\addplot [ultra  thick, color4, dashdotted]
table {%
0.0001	0.46191
0.0002	0.28157
0.0003	0.22138
0.0004	0.19124
0.0005	0.17312
0.0006	0.16101
0.0007	0.15234
0.0008	0.14582
0.0009	0.14073
0.001	0.13665
0.0011	0.1333
0.0012	0.13051
0.0013	0.12813
0.0014	0.12608
0.0015	0.12431
0.0016	0.12275
0.0017	0.12136
0.0018	0.12013
0.0019	0.11902
0.002	0.11802
0.0021	0.11712
0.0022	0.11629
0.0023	0.11553
0.0024	0.11483
0.0025	0.11419
0.0026	0.11359
0.0027	0.11303
0.0028	0.11252
0.0029	0.11203
0.003	0.11158
0.0031	0.11115
0.0032	0.11075
0.0033	0.11038
0.0034	0.11002
0.0035	0.10968
0.0036	0.10936
0.0037	0.10906
0.0038	0.10877
0.0039	0.1085
0.004	0.10823
0.0041	0.10798
0.0042	0.10775
0.0043	0.10752
0.0044	0.1073
0.0045	0.10709
0.0046	0.10689
0.0047	0.10669
0.0048	0.10651
0.0049	0.10633
0.005	0.10615
0.0051	0.10599
0.0052	0.10583
0.0053	0.10567
0.0054	0.10552
0.0055	0.10538
0.0056	0.10524
0.0057	0.1051
0.0058	0.10497
0.0059	0.10485
0.006	0.10472
0.0061	0.1046
0.0062	0.10449
0.0063	0.10437
0.0064	0.10426
0.0065	0.10416
0.0066	0.10405
0.0067	0.10395
0.0068	0.10385
0.0069	0.10376
0.007	0.10367
0.0071	0.10358
0.0072	0.10349
0.0073	0.1034
0.0074	0.10332
0.0075	0.10323
0.0076	0.10315
0.0077	0.10307
0.0078	0.103
0.0079	0.10292
0.008	0.10285
0.0081	0.10278
0.0082	0.10271
0.0083	0.10264
0.0084	0.10257
0.0085	0.10251
0.0086	0.10244
0.0087	0.10238
0.0088	0.10232
0.0089	0.10226
0.009	0.1022
0.0091	0.10214
0.0092	0.10208
0.0093	0.10203
0.0094	0.10197
0.0095	0.10192
0.0096	0.10187
0.0097	0.10181
0.0098	0.10176
0.0099	0.10171
0.01	0.10166
};
\addlegendentry{Upper bound}

\end{axis}

\end{tikzpicture}

%% file: tikz_figs/case_1_fr_latency.tex
\begin{tikzpicture}

\begin{axis}[
width=\sfwidth,
height=\sfheight,
legend cell align={left},
legend style={at={(0.9,0.2)}, anchor=south east, fill opacity=0.8, legend columns=1, draw opacity=1, text opacity=1, font=\tiny, draw=white!80!black},
tick align=outside,
tick pos=left,
x grid style={white!69.0196078431373!black},
xlabel={Latency (ms)},
xtick style={color=black},
ymin=0, ymax=1,
y grid style={white!69.0196078431373!black},
ylabel={CDF},
xtick={0,0.1,0.2,0.3},
xticklabels={0,100,200,300},
xmin=0, xmax=0.3,
ytick style={color=black},
xmajorgrids,
ymajorgrids
]
\addplot [ultra thick, color2, dashed]
table {%
0	0
0.001	0
0.002	0
0.003	0
0.004	0
0.005	0
0.006	0
0.007	9.7783e-288
0.008	6.832e-57
0.009	1
0.01	1
0.011	1
0.012	1
0.013	1
0.014	1
0.015	1
0.016	1
0.017	1
0.018	1
0.019	1
0.02	1
0.021	1
0.022	1
0.023	1
0.024	1
0.025	1
0.026	1
0.027	1
0.028	1
0.029	1
0.03	1
0.031	1
0.032	1
0.033	1
0.034	1
0.035	1
0.036	1
0.037	1
0.038	1
0.039	1
0.04	1
0.041	1
0.042	1
0.043	1
0.044	1
0.045	1
0.046	1
0.047	1
0.048	1
0.049	1
0.05	1
0.051	1
0.052	1
0.053	1
0.054	1
0.055	1
0.056	1
0.057	1
0.058	1
0.059	1
0.06	1
0.061	1
0.062	1
0.063	1
0.064	1
0.065	1
0.066	1
0.067	1
0.068	1
0.069	1
0.07	1
0.071	1
0.072	1
0.073	1
0.074	1
0.075	1
0.076	1
0.077	1
0.078	1
0.079	1
0.08	1
0.081	1
0.082	1
0.083	1
0.084	1
0.085	1
0.086	1
0.087	1
0.088	1
0.089	1
0.09	1
0.091	1
0.092	1
0.093	1
0.094	1
0.095	1
0.096	1
0.097	1
0.098	1
0.099	1
0.1	1
0.101	1
0.102	1
0.103	1
0.104	1
0.105	1
0.106	1
0.107	1
0.108	1
0.109	1
0.11	1
0.111	1
0.112	1
0.113	1
0.114	1
0.115	1
0.116	1
0.117	1
0.118	1
0.119	1
0.12	1
0.121	1
0.122	1
0.123	1
0.124	1
0.125	1
0.126	1
0.127	1
0.128	1
0.129	1
0.13	1
0.131	1
0.132	1
0.133	1
0.134	1
0.135	1
0.136	1
0.137	1
0.138	1
0.139	1
0.14	1
0.141	1
0.142	1
0.143	1
0.144	1
0.145	1
0.146	1
0.147	1
0.148	1
0.149	1
0.15	1
0.151	1
0.152	1
0.153	1
0.154	1
0.155	1
0.156	1
0.157	1
0.158	1
0.159	1
0.16	1
0.161	1
0.162	1
0.163	1
0.164	1
0.165	1
0.166	1
0.167	1
0.168	1
0.169	1
0.17	1
0.171	1
0.172	1
0.173	1
0.174	1
0.175	1
0.176	1
0.177	1
0.178	1
0.179	1
0.18	1
0.181	1
0.182	1
0.183	1
0.184	1
0.185	1
0.186	1
0.187	1
0.188	1
0.189	1
0.19	1
0.191	1
0.192	1
0.193	1
0.194	1
0.195	1
0.196	1
0.197	1
0.198	1
0.199	1
0.2	1
0.201	1
0.202	1
0.203	1
0.204	1
0.205	1
0.206	1
0.207	1
0.208	1
0.209	1
0.21	1
0.211	1
0.212	1
0.213	1
0.214	1
0.215	1
0.216	1
0.217	1
0.218	1
0.219	1
0.22	1
0.221	1
0.222	1
0.223	1
0.224	1
0.225	1
0.226	1
0.227	1
0.228	1
0.229	1
0.23	1
0.231	1
0.232	1
0.233	1
0.234	1
0.235	1
0.236	1
0.237	1
0.238	1
0.239	1
0.24	1
0.241	1
0.242	1
0.243	1
0.244	1
0.245	1
0.246	1
0.247	1
0.248	1
0.249	1
0.25	1
0.251	1
0.252	1
0.253	1
0.254	1
0.255	1
0.256	1
0.257	1
0.258	1
0.259	1
0.26	1
0.261	1
0.262	1
0.263	1
0.264	1
0.265	1
0.266	1
0.267	1
0.268	1
0.269	1
0.27	1
0.271	1
0.272	1
0.273	1
0.274	1
0.275	1
0.276	1
0.277	1
0.278	1
0.279	1
0.28	1
0.281	1
0.282	1
0.283	1
0.284	1
0.285	1
0.286	1
0.287	1
0.288	1
0.289	1
0.29	1
0.291	1
0.292	1
0.293	1
0.294	1
0.295	1
0.296	1
0.297	1
0.298	1
0.299	1
0.3	1
};
\addlegendentry{$\epsilon=10^{-2}$}

\addplot [ultra thick, color3,dashdotted]
table {%
0	0
0.001	0
0.002	0
0.003	0
0.004	0
0.005	0
0.006	0
0.007	0
0.008	0
0.009	0
0.01	0
0.011	0
0.012	0
0.013	0
0.014	0
0.015	0
0.016	0
0.017	0
0.018	0
0.019	0
0.02	0
0.021	0
0.022	0
0.023	0
0.024	0
0.025	0
0.026	0
0.027	0
0.028	0
0.029	0
0.03	0
0.031	0
0.032	0
0.033	0
0.034	9.3487e-151
0.035	2.0133e-22
0.036	1
0.037	1
0.038	1
0.039	1
0.04	1
0.041	1
0.042	1
0.043	1
0.044	1
0.045	1
0.046	1
0.047	1
0.048	1
0.049	1
0.05	1
0.051	1
0.052	1
0.053	1
0.054	1
0.055	1
0.056	1
0.057	1
0.058	1
0.059	1
0.06	1
0.061	1
0.062	1
0.063	1
0.064	1
0.065	1
0.066	1
0.067	1
0.068	1
0.069	1
0.07	1
0.071	1
0.072	1
0.073	1
0.074	1
0.075	1
0.076	1
0.077	1
0.078	1
0.079	1
0.08	1
0.081	1
0.082	1
0.083	1
0.084	1
0.085	1
0.086	1
0.087	1
0.088	1
0.089	1
0.09	1
0.091	1
0.092	1
0.093	1
0.094	1
0.095	1
0.096	1
0.097	1
0.098	1
0.099	1
0.1	1
0.101	1
0.102	1
0.103	1
0.104	1
0.105	1
0.106	1
0.107	1
0.108	1
0.109	1
0.11	1
0.111	1
0.112	1
0.113	1
0.114	1
0.115	1
0.116	1
0.117	1
0.118	1
0.119	1
0.12	1
0.121	1
0.122	1
0.123	1
0.124	1
0.125	1
0.126	1
0.127	1
0.128	1
0.129	1
0.13	1
0.131	1
0.132	1
0.133	1
0.134	1
0.135	1
0.136	1
0.137	1
0.138	1
0.139	1
0.14	1
0.141	1
0.142	1
0.143	1
0.144	1
0.145	1
0.146	1
0.147	1
0.148	1
0.149	1
0.15	1
0.151	1
0.152	1
0.153	1
0.154	1
0.155	1
0.156	1
0.157	1
0.158	1
0.159	1
0.16	1
0.161	1
0.162	1
0.163	1
0.164	1
0.165	1
0.166	1
0.167	1
0.168	1
0.169	1
0.17	1
0.171	1
0.172	1
0.173	1
0.174	1
0.175	1
0.176	1
0.177	1
0.178	1
0.179	1
0.18	1
0.181	1
0.182	1
0.183	1
0.184	1
0.185	1
0.186	1
0.187	1
0.188	1
0.189	1
0.19	1
0.191	1
0.192	1
0.193	1
0.194	1
0.195	1
0.196	1
0.197	1
0.198	1
0.199	1
0.2	1
0.201	1
0.202	1
0.203	1
0.204	1
0.205	1
0.206	1
0.207	1
0.208	1
0.209	1
0.21	1
0.211	1
0.212	1
0.213	1
0.214	1
0.215	1
0.216	1
0.217	1
0.218	1
0.219	1
0.22	1
0.221	1
0.222	1
0.223	1
0.224	1
0.225	1
0.226	1
0.227	1
0.228	1
0.229	1
0.23	1
0.231	1
0.232	1
0.233	1
0.234	1
0.235	1
0.236	1
0.237	1
0.238	1
0.239	1
0.24	1
0.241	1
0.242	1
0.243	1
0.244	1
0.245	1
0.246	1
0.247	1
0.248	1
0.249	1
0.25	1
0.251	1
0.252	1
0.253	1
0.254	1
0.255	1
0.256	1
0.257	1
0.258	1
0.259	1
0.26	1
0.261	1
0.262	1
0.263	1
0.264	1
0.265	1
0.266	1
0.267	1
0.268	1
0.269	1
0.27	1
0.271	1
0.272	1
0.273	1
0.274	1
0.275	1
0.276	1
0.277	1
0.278	1
0.279	1
0.28	1
0.281	1
0.282	1
0.283	1
0.284	1
0.285	1
0.286	1
0.287	1
0.288	1
0.289	1
0.29	1
0.291	1
0.292	1
0.293	1
0.294	1
0.295	1
0.296	1
0.297	1
0.298	1
0.299	1
0.3	1
};
\addlegendentry{$\epsilon=10^{-3}$}

\addplot [ultra thick, color5,densely dotted]
table {%
0	0
0.001	0
0.002	0
0.003	0
0.004	0
0.005	0
0.006	0
0.007	0
0.008	0
0.009	0
0.01	0
0.011	0
0.012	0
0.013	0
0.014	0
0.015	0
0.016	0
0.017	0
0.018	0
0.019	0
0.02	0
0.021	0
0.022	0
0.023	0
0.024	0
0.025	0
0.026	0
0.027	0
0.028	0
0.029	0
0.03	0
0.031	0
0.032	0
0.033	0
0.034	0
0.035	0
0.036	0
0.037	0
0.038	0
0.039	0
0.04	0
0.041	0
0.042	0
0.043	0
0.044	0
0.045	0
0.046	0
0.047	0
0.048	0
0.049	0
0.05	0
0.051	0
0.052	0
0.053	0
0.054	0
0.055	0
0.056	0
0.057	0
0.058	0
0.059	0
0.06	0
0.061	0
0.062	0
0.063	0
0.064	0
0.065	0
0.066	0
0.067	0
0.068	0
0.069	0
0.07	0
0.071	0
0.072	0
0.073	0
0.074	0
0.075	0
0.076	0
0.077	0
0.078	0
0.079	0
0.08	0
0.081	0
0.082	0
0.083	0
0.084	0
0.085	0
0.086	0
0.087	0
0.088	0
0.089	0
0.09	0
0.091	0
0.092	0
0.093	0
0.094	0
0.095	0
0.096	0
0.097	0
0.098	0
0.099	0
0.1	0
0.101	0
0.102	0
0.103	0
0.104	0
0.105	0
0.106	0
0.107	0
0.108	0
0.109	0
0.11	0
0.111	0
0.112	0
0.113	0
0.114	0
0.115	0
0.116	0
0.117	0
0.118	0
0.119	0
0.12	0
0.121	0
0.122	0
0.123	0
0.124	0
0.125	0
0.126	0
0.127	0
0.128	0
0.129	0
0.13	0
0.131	0
0.132	0
0.133	0
0.134	0
0.135	0
0.136	0
0.137	0
0.138	0
0.139	0
0.14	0
0.141	0
0.142	0
0.143	0
0.144	0
0.145	0
0.146	0
0.147	0
0.148	0
0.149	0
0.15	0
0.151	0
0.152	0
0.153	0
0.154	0
0.155	0
0.156	0
0.157	0
0.158	0
0.159	0
0.16	0
0.161	0
0.162	0
0.163	0
0.164	0
0.165	0
0.166	0
0.167	0
0.168	0
0.169	0
0.17	0
0.171	0
0.172	0
0.173	0
0.174	0
0.175	0
0.176	0
0.177	0
0.178	0
0.179	0
0.18	0
0.181	0
0.182	0
0.183	0
0.184	0
0.185	0
0.186	0
0.187	0
0.188	0
0.189	0
0.19	0
0.191	0
0.192	0
0.193	0
0.194	0
0.195	0
0.196	0
0.197	0
0.198	0
0.199	0
0.2	0
0.201	0
0.202	0
0.203	0
0.204	0
0.205	0
0.206	0
0.207	0
0.208	0
0.209	0
0.21	0
0.211	0
0.212	0
0.213	0
0.214	0
0.215	0
0.216	0
0.217	0
0.218	0
0.219	0
0.22	0
0.221	0
0.222	0
0.223	0
0.224	0
0.225	0
0.226	0
0.227	0
0.228	0
0.229	0
0.23	0
0.231	0
0.232	0
0.233	0
0.234	0
0.235	0
0.236	0
0.237	0
0.238	0
0.239	0
0.24	0
0.241	0
0.242	0
0.243	0
0.244	0
0.245	0
0.246	0
0.247	0
0.248	0
0.249	0
0.25	0
0.251	0
0.252	0
0.253	0
0.254	0
0.255	0
0.256	0
0.257	0
0.258	0
0.259	0
0.26	0
0.261	0
0.262	0
0.263	0
0.264	0
0.265	0
0.266	0
0.267	0
0.268	0
0.269	0
0.27	0
0.271	0
0.272	0
0.273	0
0.274	0
0.275	0
0.276	0
0.277	0
0.278	0
0.279	0
0.28	1.3959e-265
0.281	8.6261e-180
0.282	1.1677e-110
0.283	3.4621e-58
0.284	2.2485e-22
0.285	0.0031988
0.286	1
0.287	1
0.288	1
0.289	1
0.29	1
0.291	1
0.292	1
0.293	1
0.294	1
0.295	1
0.296	1
0.297	1
0.298	1
0.299	1
0.3	1
};
\addlegendentry{$\epsilon=10^{-4}$}

\end{axis}

\end{tikzpicture}

\vspace{-0.3cm}

%% file: tikz_figs/case_1_fr_compr.tex
\begin{tikzpicture}

\begin{axis}[
width=\sfwidth,
height=\sfheight,
legend cell align={left},
legend style={at={(0.995,0.01)}, anchor=south east, fill opacity=0.8, legend columns=1, draw opacity=1, text opacity=1, font=\tiny, draw=white!80!black},
tick align=outside,
tick pos=left,
x grid style={white!69.0196078431373!black},
xlabel={Latency (ms)},
xtick style={color=black},
ymin=0, ymax=1,
y grid style={white!69.0196078431373!black},
ylabel={CDF},
xmin=0, xmax=0.1,
xtick={0,0,0.02,0.04,0.06,0.08,0.1},
xticklabels={0,0,20,40,60,80,100},
ytick style={color=black},
xmajorgrids,
ymajorgrids
]

\addplot [ultra thick, color2, dashed]
table {%
0	0
0.001	0.55422
0.002	0.82016
0.003	0.93067
0.004	0.97397
0.005	0.99039
0.006	0.9965
0.007	0.99873
0.008	0.99955
0.009	0.99984
0.01	0.99994
0.011	0.99998
0.012	0.99999
0.013	1
0.014	1
0.015	1
0.016	1
0.017	1
0.018	1
0.019	1
0.02	1
0.021	1
0.022	1
0.023	1
0.024	1
0.025	1
0.026	1
0.027	1
0.028	1
0.029	1
0.03	1
0.031	1
0.032	1
0.033	1
0.034	1
0.035	1
0.036	1
0.037	1
0.038	1
0.039	1
0.04	1
0.041	1
0.042	1
0.043	1
0.044	1
0.045	1
0.046	1
0.047	1
0.048	1
0.049	1
0.05	1
0.051	1
0.052	1
0.053	1
0.054	1
0.055	1
0.056	1
0.057	1
0.058	1
0.059	1
0.06	1
0.061	1
0.062	1
0.063	1
0.064	1
0.065	1
0.066	1
0.067	1
0.068	1
0.069	1
0.07	1
0.071	1
0.072	1
0.073	1
0.074	1
0.075	1
0.076	1
0.077	1
0.078	1
0.079	1
0.08	1
0.081	1
0.082	1
0.083	1
0.084	1
0.085	1
0.086	1
0.087	1
0.088	1
0.089	1
0.09	1
0.091	1
0.092	1
0.093	1
0.094	1
0.095	1
0.096	1
0.097	1
0.098	1
0.099	1
0.1	1
};
\addlegendentry{$Q=1.1$}

\addplot [ultra thick, color3,dashdotted]
table {%
0	0
0.001	0.10946
0.002	0.23076
0.003	0.34718
0.004	0.45252
0.005	0.54481
0.006	0.62403
0.007	0.69109
0.008	0.74726
0.009	0.79395
0.01	0.83251
0.011	0.86421
0.012	0.89015
0.013	0.91131
0.014	0.92851
0.015	0.94247
0.016	0.95376
0.017	0.96289
0.018	0.97024
0.019	0.97616
0.02	0.98093
0.021	0.98475
0.022	0.98782
0.023	0.99027
0.024	0.99224
0.025	0.99381
0.026	0.99507
0.027	0.99607
0.028	0.99687
0.029	0.99751
0.03	0.99802
0.031	0.99843
0.032	0.99875
0.033	0.99901
0.034	0.99921
0.035	0.99937
0.036	0.9995
0.037	0.99961
0.038	0.99969
0.039	0.99975
0.04	0.9998
0.041	0.99984
0.042	0.99988
0.043	0.9999
0.044	0.99992
0.045	0.99994
0.046	0.99995
0.047	0.99996
0.048	0.99997
0.049	0.99998
0.05	0.99998
0.051	0.99999
0.052	0.99999
0.053	0.99999
0.054	0.99999
0.055	0.99999
0.056	1
0.057	1
0.058	1
0.059	1
0.06	1
0.061	1
0.062	1
0.063	1
0.064	1
0.065	1
0.066	1
0.067	1
0.068	1
0.069	1
0.07	1
0.071	1
0.072	1
0.073	1
0.074	1
0.075	1
0.076	1
0.077	1
0.078	1
0.079	1
0.08	1
0.081	1
0.082	1
0.083	1
0.084	1
0.085	1
0.086	1
0.087	1
0.088	1
0.089	1
0.09	1
0.091	1
0.092	1
0.093	1
0.094	1
0.095	1
0.096	1
0.097	1
0.098	1
0.099	1
0.1	1
};
\addlegendentry{$Q=1.3$}

\addplot [ultra thick, color5,densely dotted]
table {%
0	0
0.001	0.03038
0.002	0.069439
0.003	0.11293
0.004	0.15858
0.005	0.20499
0.006	0.25121
0.007	0.29659
0.008	0.34069
0.009	0.38323
0.01	0.42398
0.011	0.46285
0.012	0.49975
0.013	0.53467
0.014	0.56761
0.015	0.59861
0.016	0.62772
0.017	0.655
0.018	0.68052
0.019	0.70436
0.02	0.72659
0.021	0.7473
0.022	0.76657
0.023	0.78449
0.024	0.80112
0.025	0.81656
0.026	0.83087
0.027	0.84413
0.028	0.85641
0.029	0.86777
0.03	0.87827
0.031	0.88797
0.032	0.89694
0.033	0.90522
0.034	0.91285
0.035	0.9199
0.036	0.92639
0.037	0.93238
0.038	0.9379
0.039	0.94298
0.04	0.94766
0.041	0.95196
0.042	0.95592
0.043	0.95957
0.044	0.96292
0.045	0.966
0.046	0.96883
0.047	0.97143
0.048	0.97382
0.049	0.97602
0.05	0.97804
0.051	0.97989
0.052	0.98159
0.053	0.98315
0.054	0.98458
0.055	0.9859
0.056	0.9871
0.057	0.98821
0.058	0.98922
0.059	0.99015
0.06	0.991
0.061	0.99178
0.062	0.9925
0.063	0.99316
0.064	0.99376
0.065	0.99431
0.066	0.99481
0.067	0.99528
0.068	0.9957
0.069	0.99609
0.07	0.99644
0.071	0.99677
0.072	0.99707
0.073	0.99734
0.074	0.99759
0.075	0.99782
0.076	0.99802
0.077	0.99822
0.078	0.99839
0.079	0.99855
0.08	0.9987
0.081	0.99883
0.082	0.99895
0.083	0.99907
0.084	0.99917
0.085	0.99926
0.086	0.99935
0.087	0.99943
0.088	0.9995
0.089	0.99956
0.09	0.99962
0.091	0.99968
0.092	0.99973
0.093	0.99977
0.094	0.99981
0.095	0.99985
0.096	0.99989
0.097	0.99992
0.098	0.99995
0.099	0.99998
0.1	1
};
\addlegendentry{$Q=1.5$}

\end{axis}

\end{tikzpicture}

\vspace{-0.3cm}

%% file: tikz_figs/topt_eps.tex
\begin{tikzpicture}

\begin{axis}[
width=\sfwidth,
height=\sfheight,
xmode=log,
legend cell align={left},
legend style={at={(0.995,0.99)}, anchor=north east, fill opacity=0.8, legend columns=1, draw opacity=1, text opacity=1, font=\tiny, draw=white!80!black},
tick align=outside,
tick pos=left,
x grid style={white!69.0196078431373!black},
ylabel={ Optimal closed-loop latency (ms) }, 
ymin=0, ymax=0.5,
ytick={0,0.1,0.2,0.3,0.4,0.5},
yticklabels={0,100,200,300,400,500},
xtick style={color=black},
y grid style={white!69.0196078431373!black},
xlabel={$\epsilon$},
xmin=1e-4, xmax=1e-2,
ytick style={color=black},
xmajorgrids,
ymajorgrids
]

\addplot [ultra thick, color2, dashed]
table {%
0.000100000000000000	0.427000000000000
0.000200000000000000	0.266000000000000
0.000300000000000000	0.209000000000000
0.000400000000000000	0.181000000000000
0.000500000000000000	0.163000000000000
0.000600000000000000	0.151000000000000
0.000700000000000000	0.143000000000000
0.000800000000000000	0.137000000000000
0.000900000000000000	0.132000000000000
0.00100000000000000	0.128000000000000
0.00110000000000000	0.124000000000000
0.00120000000000000	0.121000000000000
0.00130000000000000	0.120000000000000
0.00140000000000000	0.117000000000000
0.00150000000000000	0.116000000000000
0.00160000000000000	0.115000000000000
0.00170000000000000	0.113000000000000
0.00180000000000000	0.112000000000000
0.00190000000000000	0.111000000000000
0.00200000000000000	0.109000000000000
0.00210000000000000	0.108000000000000
0.00220000000000000	0.108000000000000
0.00230000000000000	0.107000000000000
0.00240000000000000	0.107000000000000
0.00250000000000000	0.106000000000000
0.00260000000000000	0.105000000000000
0.00270000000000000	0.104000000000000
0.00280000000000000	0.104000000000000
0.00290000000000000	0.103000000000000
0.00300000000000000	0.103000000000000
0.00310000000000000	0.103000000000000
0.00320000000000000	0.103000000000000
0.00330000000000000	0.102000000000000
0.00340000000000000	0.102000000000000
0.00350000000000000	0.101000000000000
0.00360000000000000	0.100000000000000
0.00370000000000000	0.100000000000000
0.00380000000000000	0.100000000000000
0.00390000000000000	0.100000000000000
0.00400000000000000	0.100000000000000
0.00410000000000000	0.0990000000000000
0.00420000000000000	0.0990000000000000
0.00430000000000000	0.0990000000000000
0.00440000000000000	0.0990000000000000
0.00450000000000000	0.0990000000000000
0.00460000000000000	0.0990000000000000
0.00470000000000000	0.0980000000000000
0.00480000000000000	0.0980000000000000
0.00490000000000000	0.0980000000000000
0.00500000000000000	0.0980000000000000
0.00510000000000000	0.0980000000000000
0.00520000000000000	0.0980000000000000
0.00530000000000000	0.0980000000000000
0.00540000000000000	0.0980000000000000
0.00550000000000000	0.0970000000000000
0.00560000000000000	0.0960000000000000
0.00570000000000000	0.0960000000000000
0.00580000000000000	0.0960000000000000
0.00590000000000000	0.0960000000000000
0.00600000000000000	0.0960000000000000
0.00610000000000000	0.0960000000000000
0.00620000000000000	0.0960000000000000
0.00630000000000000	0.0960000000000000
0.00640000000000000	0.0960000000000000
0.00650000000000000	0.0960000000000000
0.00660000000000000	0.0950000000000000
0.00670000000000000	0.0950000000000000
0.00680000000000000	0.0950000000000000
0.00690000000000000	0.0950000000000000
0.00700000000000000	0.0950000000000000
0.00710000000000000	0.0950000000000000
0.00720000000000000	0.0950000000000000
0.00730000000000000	0.0950000000000000
0.00740000000000000	0.0950000000000000
0.00750000000000000	0.0950000000000000
0.00760000000000000	0.0950000000000000
0.00770000000000000	0.0950000000000000
0.00780000000000000	0.0950000000000000
0.00790000000000000	0.0950000000000000
0.00800000000000000	0.0950000000000000
0.00810000000000000	0.0940000000000000
0.00820000000000000	0.0940000000000000
0.00830000000000000	0.0940000000000000
0.00840000000000000	0.0940000000000000
0.00850000000000000	0.0940000000000000
0.00860000000000000	0.0940000000000000
0.00870000000000000	0.0940000000000000
0.00880000000000000	0.0940000000000000
0.00890000000000000	0.0940000000000000
0.00900000000000000	0.0940000000000000
0.00910000000000000	0.0940000000000000
0.00920000000000000	0.0940000000000000
0.00930000000000000	0.0940000000000000
0.00940000000000000	0.0940000000000000
0.00950000000000000	0.0940000000000000
0.00960000000000000	0.0940000000000000
0.00970000000000000	0.0940000000000000
0.00980000000000000	0.0940000000000000
0.00990000000000000	0.0940000000000000
0.0100000000000000	0.0940000000000000
};
\addlegendentry{Case 1, $f_R=1$ GHz}

\addplot [ultra thick, color3,dashdotted]
table {%
0.000100000000000000	0.393000000000000
0.000200000000000000	0.249000000000000
0.000300000000000000	0.200000000000000
0.000400000000000000	0.174000000000000
0.000500000000000000	0.159000000000000
0.000600000000000000	0.148000000000000
0.000700000000000000	0.140000000000000
0.000800000000000000	0.135000000000000
0.000900000000000000	0.130000000000000
0.00100000000000000	0.126000000000000
0.00110000000000000	0.123000000000000
0.00120000000000000	0.120000000000000
0.00130000000000000	0.118000000000000
0.00140000000000000	0.116000000000000
0.00150000000000000	0.115000000000000
0.00160000000000000	0.113000000000000
0.00170000000000000	0.111000000000000
0.00180000000000000	0.111000000000000
0.00190000000000000	0.110000000000000
0.00200000000000000	0.108000000000000
0.00210000000000000	0.107000000000000
0.00220000000000000	0.107000000000000
0.00230000000000000	0.106000000000000
0.00240000000000000	0.106000000000000
0.00250000000000000	0.105000000000000
0.00260000000000000	0.104000000000000
0.00270000000000000	0.103000000000000
0.00280000000000000	0.103000000000000
0.00290000000000000	0.103000000000000
0.00300000000000000	0.102000000000000
0.00310000000000000	0.102000000000000
0.00320000000000000	0.102000000000000
0.00330000000000000	0.101000000000000
0.00340000000000000	0.101000000000000
0.00350000000000000	0.101000000000000
0.00360000000000000	0.100000000000000
0.00370000000000000	0.0990000000000000
0.00380000000000000	0.0990000000000000
0.00390000000000000	0.0990000000000000
0.00400000000000000	0.0990000000000000
0.00410000000000000	0.0990000000000000
0.00420000000000000	0.0980000000000000
0.00430000000000000	0.0980000000000000
0.00440000000000000	0.0980000000000000
0.00450000000000000	0.0980000000000000
0.00460000000000000	0.0980000000000000
0.00470000000000000	0.0980000000000000
0.00480000000000000	0.0970000000000000
0.00490000000000000	0.0970000000000000
0.00500000000000000	0.0970000000000000
0.00510000000000000	0.0970000000000000
0.00520000000000000	0.0970000000000000
0.00530000000000000	0.0970000000000000
0.00540000000000000	0.0970000000000000
0.00550000000000000	0.0970000000000000
0.00560000000000000	0.0950000000000000
0.00570000000000000	0.0950000000000000
0.00580000000000000	0.0950000000000000
0.00590000000000000	0.0950000000000000
0.00600000000000000	0.0950000000000000
0.00610000000000000	0.0950000000000000
0.00620000000000000	0.0950000000000000
0.00630000000000000	0.0950000000000000
0.00640000000000000	0.0950000000000000
0.00650000000000000	0.0950000000000000
0.00660000000000000	0.0950000000000000
0.00670000000000000	0.0940000000000000
0.00680000000000000	0.0940000000000000
0.00690000000000000	0.0940000000000000
0.00700000000000000	0.0940000000000000
0.00710000000000000	0.0940000000000000
0.00720000000000000	0.0940000000000000
0.00730000000000000	0.0940000000000000
0.00740000000000000	0.0940000000000000
0.00750000000000000	0.0940000000000000
0.00760000000000000	0.0940000000000000
0.00770000000000000	0.0940000000000000
0.00780000000000000	0.0940000000000000
0.00790000000000000	0.0940000000000000
0.00800000000000000	0.0940000000000000
0.00810000000000000	0.0940000000000000
0.00820000000000000	0.0940000000000000
0.00830000000000000	0.0930000000000000
0.00840000000000000	0.0930000000000000
0.00850000000000000	0.0930000000000000
0.00860000000000000	0.0930000000000000
0.00870000000000000	0.0930000000000000
0.00880000000000000	0.0930000000000000
0.00890000000000000	0.0930000000000000
0.00900000000000000	0.0930000000000000
0.00910000000000000	0.0930000000000000
0.00920000000000000	0.0930000000000000
0.00930000000000000	0.0930000000000000
0.00940000000000000	0.0930000000000000
0.00950000000000000	0.0930000000000000
0.00960000000000000	0.0930000000000000
0.00970000000000000	0.0930000000000000
0.00980000000000000	0.0930000000000000
0.00990000000000000	0.0930000000000000
0.0100000000000000	0.0930000000000000
};
\addlegendentry{Case 1, $f_R=3$ GHz}

\addplot [ultra thick, color5,densely dotted]
table {%
0.000100000000000000	0.378000000000000
0.000200000000000000	0.241000000000000
0.000300000000000000	0.194000000000000
0.000400000000000000	0.170000000000000
0.000500000000000000	0.155000000000000
0.000600000000000000	0.145000000000000
0.000700000000000000	0.138000000000000
0.000800000000000000	0.133000000000000
0.000900000000000000	0.128000000000000
0.00100000000000000	0.125000000000000
0.00110000000000000	0.122000000000000
0.00120000000000000	0.119000000000000
0.00130000000000000	0.117000000000000
0.00140000000000000	0.115000000000000
0.00150000000000000	0.114000000000000
0.00160000000000000	0.113000000000000
0.00170000000000000	0.111000000000000
0.00180000000000000	0.110000000000000
0.00190000000000000	0.109000000000000
0.00200000000000000	0.107000000000000
0.00210000000000000	0.107000000000000
0.00220000000000000	0.106000000000000
0.00230000000000000	0.106000000000000
0.00240000000000000	0.105000000000000
0.00250000000000000	0.105000000000000
0.00260000000000000	0.103000000000000
0.00270000000000000	0.103000000000000
0.00280000000000000	0.102000000000000
0.00290000000000000	0.102000000000000
0.00300000000000000	0.102000000000000
0.00310000000000000	0.102000000000000
0.00320000000000000	0.101000000000000
0.00330000000000000	0.101000000000000
0.00340000000000000	0.101000000000000
0.00350000000000000	0.100000000000000
0.00360000000000000	0.0990000000000000
0.00370000000000000	0.0990000000000000
0.00380000000000000	0.0990000000000000
0.00390000000000000	0.0990000000000000
0.00400000000000000	0.0990000000000000
0.00410000000000000	0.0980000000000000
0.00420000000000000	0.0980000000000000
0.00430000000000000	0.0980000000000000
0.00440000000000000	0.0980000000000000
0.00450000000000000	0.0980000000000000
0.00460000000000000	0.0980000000000000
0.00470000000000000	0.0970000000000000
0.00480000000000000	0.0970000000000000
0.00490000000000000	0.0970000000000000
0.00500000000000000	0.0970000000000000
0.00510000000000000	0.0970000000000000
0.00520000000000000	0.0970000000000000
0.00530000000000000	0.0970000000000000
0.00540000000000000	0.0970000000000000
0.00550000000000000	0.0970000000000000
0.00560000000000000	0.0950000000000000
0.00570000000000000	0.0950000000000000
0.00580000000000000	0.0950000000000000
0.00590000000000000	0.0950000000000000
0.00600000000000000	0.0950000000000000
0.00610000000000000	0.0950000000000000
0.00620000000000000	0.0950000000000000
0.00630000000000000	0.0950000000000000
0.00640000000000000	0.0950000000000000
0.00650000000000000	0.0950000000000000
0.00660000000000000	0.0950000000000000
0.00670000000000000	0.0940000000000000
0.00680000000000000	0.0940000000000000
0.00690000000000000	0.0940000000000000
0.00700000000000000	0.0940000000000000
0.00710000000000000	0.0940000000000000
0.00720000000000000	0.0940000000000000
0.00730000000000000	0.0940000000000000
0.00740000000000000	0.0940000000000000
0.00750000000000000	0.0940000000000000
0.00760000000000000	0.0940000000000000
0.00770000000000000	0.0940000000000000
0.00780000000000000	0.0940000000000000
0.00790000000000000	0.0940000000000000
0.00800000000000000	0.0940000000000000
0.00810000000000000	0.0940000000000000
0.00820000000000000	0.0940000000000000
0.00830000000000000	0.0930000000000000
0.00840000000000000	0.0930000000000000
0.00850000000000000	0.0930000000000000
0.00860000000000000	0.0930000000000000
0.00870000000000000	0.0930000000000000
0.00880000000000000	0.0930000000000000
0.00890000000000000	0.0930000000000000
0.00900000000000000	0.0930000000000000
0.00910000000000000	0.0930000000000000
0.00920000000000000	0.0930000000000000
0.00930000000000000	0.0930000000000000
0.00940000000000000	0.0930000000000000
0.00950000000000000	0.0930000000000000
0.00960000000000000	0.0930000000000000
0.00970000000000000	0.0930000000000000
0.00980000000000000	0.0930000000000000
0.00990000000000000	0.0930000000000000
0.0100000000000000	0.0930000000000000
};
\addlegendentry{Case 1, $f_R=5$ GHz}

\addplot[thick, color6]
table {%
0.0001	0.451
0.0002	0.27
0.0003	0.21
0.0004	0.18
0.0005	0.162
0.0006	0.149
0.0007	0.141
0.0008	0.135
0.0009	0.13
0.001	0.126
0.0011	0.122
0.0012	0.119
0.0013	0.117
0.0014	0.115
0.0015	0.113
0.0016	0.112
0.0017	0.11
0.0018	0.109
0.0019	0.108
0.002	0.106
0.0021	0.106
0.0022	0.105
0.0023	0.104
0.0024	0.104
0.0025	0.103
0.0026	0.102
0.0027	0.102
0.0028	0.101
0.0029	0.101
0.003	0.1
0.0031	0.1
0.0032	0.1
0.0033	0.1
0.0034	0.099
0.0035	0.099
0.0036	0.098
0.0037	0.098
0.0038	0.097
0.0039	0.097
0.004	0.097
0.0041	0.097
0.0042	0.097
0.0043	0.096
0.0044	0.096
0.0045	0.096
0.0046	0.096
0.0047	0.096
0.0048	0.096
0.0049	0.095
0.005	0.095
0.0051	0.095
0.0052	0.095
0.0053	0.095
0.0054	0.095
0.0055	0.095
0.0056	0.094
0.0057	0.093
0.0058	0.093
0.0059	0.093
0.006	0.093
0.0061	0.093
0.0062	0.093
0.0063	0.093
0.0064	0.093
0.0065	0.093
0.0066	0.093
0.0067	0.093
0.0068	0.093
0.0069	0.092
0.007	0.092
0.0071	0.092
0.0072	0.092
0.0073	0.092
0.0074	0.092
0.0075	0.092
0.0076	0.092
0.0077	0.092
0.0078	0.092
0.0079	0.092
0.008	0.092
0.0081	0.092
0.0082	0.092
0.0083	0.092
0.0084	0.092
0.0085	0.091
0.0086	0.091
0.0087	0.091
0.0088	0.091
0.0089	0.091
0.009	0.091
0.0091	0.091
0.0092	0.091
0.0093	0.091
0.0094	0.091
0.0095	0.091
0.0096	0.091
0.0097	0.091
0.0098	0.091
0.0099	0.091
0.01	0.091
};
\addlegendentry{Case 2}

\end{axis}

\end{tikzpicture}

\vspace{-0.3cm}

%% file: tikz_figs/topt_q.tex
\begin{tikzpicture}

\begin{axis}[
width=\sfwidth,
height=\sfheight,
xmode=log,
legend cell align={left},
legend style={at={(0.995,0.99)}, anchor=north east, fill opacity=0.8, legend columns=1, draw opacity=1, text opacity=1, font=\tiny, draw=white!80!black},
tick align=outside,
tick pos=left,
x grid style={white!69.0196078431373!black},
ylabel={ Optimal compression ratio  },
ymin=1, ymax=1.6,
xtick style={color=black},
y grid style={white!69.0196078431373!black},
xlabel={$\epsilon$},
xmin=1e-4, xmax=1e-2,
ytick style={color=black},
xmajorgrids,
ymajorgrids
]

\addplot [ultra thick, color2, dashed]
table {%
0.000100000000000000	1.19500000000000
0.000200000000000000	1.11500000000000
0.000300000000000000	1.09000000000000
0.000400000000000000	1.05500000000000
0.000500000000000000	1.02000000000000
0.000600000000000000	1.02000000000000
0.000700000000000000	1
0.000800000000000000	1
0.000900000000000000	1
0.00100000000000000	1
0.00110000000000000	1
0.00120000000000000	1
0.00130000000000000	1
0.00140000000000000	1
0.00150000000000000	1
0.00160000000000000	1
0.00170000000000000	1
0.00180000000000000	1
0.00190000000000000	1
0.00200000000000000	1
0.00210000000000000	1
0.00220000000000000	1
0.00230000000000000	1
0.00240000000000000	1
0.00250000000000000	1
0.00260000000000000	1
0.00270000000000000	1
0.00280000000000000	1
0.00290000000000000	1
0.00300000000000000	1
0.00310000000000000	1
0.00320000000000000	1
0.00330000000000000	1
0.00340000000000000	1
0.00350000000000000	1
0.00360000000000000	1
0.00370000000000000	1
0.00380000000000000	1
0.00390000000000000	1
0.00400000000000000	1
0.00410000000000000	1
0.00420000000000000	1
0.00430000000000000	1
0.00440000000000000	1
0.00450000000000000	1
0.00460000000000000	1
0.00470000000000000	1
0.00480000000000000	1
0.00490000000000000	1
0.00500000000000000	1
0.00510000000000000	1
0.00520000000000000	1
0.00530000000000000	1
0.00540000000000000	1
0.00550000000000000	1
0.00560000000000000	1
0.00570000000000000	1
0.00580000000000000	1
0.00590000000000000	1
0.00600000000000000	1
0.00610000000000000	1
0.00620000000000000	1
0.00630000000000000	1
0.00640000000000000	1
0.00650000000000000	1
0.00660000000000000	1
0.00670000000000000	1
0.00680000000000000	1
0.00690000000000000	1
0.00700000000000000	1
0.00710000000000000	1
0.00720000000000000	1
0.00730000000000000	1
0.00740000000000000	1
0.00750000000000000	1
0.00760000000000000	1
0.00770000000000000	1
0.00780000000000000	1
0.00790000000000000	1
0.00800000000000000	1
0.00810000000000000	1
0.00820000000000000	1
0.00830000000000000	1
0.00840000000000000	1
0.00850000000000000	1
0.00860000000000000	1
0.00870000000000000	1
0.00880000000000000	1
0.00890000000000000	1
0.00900000000000000	1
0.00910000000000000	1
0.00920000000000000	1
0.00930000000000000	1
0.00940000000000000	1
0.00950000000000000	1
0.00960000000000000	1
0.00970000000000000	1
0.00980000000000000	1
0.00990000000000000	1
0.0100000000000000	1
};
\addlegendentry{Case 1, $f_R=1$ GHz}

\addplot [ ultra thick, color3,dashdotted]
table {%
0.000100000000000000	1.39500000000000
0.000200000000000000	1.33000000000000
0.000300000000000000	1.25000000000000
0.000400000000000000	1.22000000000000
0.000500000000000000	1.18000000000000
0.000600000000000000	1.18000000000000
0.000700000000000000	1.14500000000000
0.000800000000000000	1.12000000000000
0.000900000000000000	1.10500000000000
0.00100000000000000	1.10500000000000
0.00110000000000000	1.10000000000000
0.00120000000000000	1.09500000000000
0.00130000000000000	1.08500000000000
0.00140000000000000	1.07500000000000
0.00150000000000000	1
0.00160000000000000	1
0.00170000000000000	1
0.00180000000000000	1
0.00190000000000000	1
0.00200000000000000	1
0.00210000000000000	1
0.00220000000000000	1
0.00230000000000000	1
0.00240000000000000	1
0.00250000000000000	1
0.00260000000000000	1
0.00270000000000000	1
0.00280000000000000	1
0.00290000000000000	1
0.00300000000000000	1
0.00310000000000000	1
0.00320000000000000	1
0.00330000000000000	1
0.00340000000000000	1
0.00350000000000000	1
0.00360000000000000	1
0.00370000000000000	1
0.00380000000000000	1
0.00390000000000000	1
0.00400000000000000	1
0.00410000000000000	1
0.00420000000000000	1
0.00430000000000000	1
0.00440000000000000	1
0.00450000000000000	1
0.00460000000000000	1
0.00470000000000000	1
0.00480000000000000	1
0.00490000000000000	1
0.00500000000000000	1
0.00510000000000000	1
0.00520000000000000	1
0.00530000000000000	1
0.00540000000000000	1
0.00550000000000000	1
0.00560000000000000	1
0.00570000000000000	1
0.00580000000000000	1
0.00590000000000000	1
0.00600000000000000	1
0.00610000000000000	1
0.00620000000000000	1
0.00630000000000000	1
0.00640000000000000	1
0.00650000000000000	1
0.00660000000000000	1
0.00670000000000000	1
0.00680000000000000	1
0.00690000000000000	1
0.00700000000000000	1
0.00710000000000000	1
0.00720000000000000	1
0.00730000000000000	1
0.00740000000000000	1
0.00750000000000000	1
0.00760000000000000	1
0.00770000000000000	1
0.00780000000000000	1
0.00790000000000000	1
0.00800000000000000	1
0.00810000000000000	1
0.00820000000000000	1
0.00830000000000000	1
0.00840000000000000	1
0.00850000000000000	1
0.00860000000000000	1
0.00870000000000000	1
0.00880000000000000	1
0.00890000000000000	1
0.00900000000000000	1
0.00910000000000000	1
0.00920000000000000	1
0.00930000000000000	1
0.00940000000000000	1
0.00950000000000000	1
0.00960000000000000	1
0.00970000000000000	1
0.00980000000000000	1
0.00990000000000000	1
0.0100000000000000	1
};
\addlegendentry{Case 1, $f_R=3$ GHz}

\addplot [ultra thick, color5,densely dotted]
table {%
0.000100000000000000	1.51000000000000
0.000200000000000000	1.41000000000000
0.000300000000000000	1.35500000000000
0.000400000000000000	1.32500000000000
0.000500000000000000	1.29500000000000
0.000600000000000000	1.26000000000000
0.000700000000000000	1.23500000000000
0.000800000000000000	1.23500000000000
0.000900000000000000	1.22000000000000
0.00100000000000000	1.19000000000000
0.00110000000000000	1.18000000000000
0.00120000000000000	1.18000000000000
0.00130000000000000	1.16500000000000
0.00140000000000000	1.16500000000000
0.00150000000000000	1.16000000000000
0.00160000000000000	1.15500000000000
0.00170000000000000	1.15500000000000
0.00180000000000000	1.14000000000000
0.00190000000000000	1.14000000000000
0.00200000000000000	1.13000000000000
0.00210000000000000	1.13500000000000
0.00220000000000000	1.12500000000000
0.00230000000000000	1.12500000000000
0.00240000000000000	1.12500000000000
0.00250000000000000	1.10500000000000
0.00260000000000000	1.10000000000000
0.00270000000000000	1.08500000000000
0.00280000000000000	1.08500000000000
0.00290000000000000	1.08500000000000
0.00300000000000000	1.08500000000000
0.00310000000000000	1.08500000000000
0.00320000000000000	1
0.00330000000000000	1
0.00340000000000000	1
0.00350000000000000	1
0.00360000000000000	1
0.00370000000000000	1
0.00380000000000000	1
0.00390000000000000	1
0.00400000000000000	1
0.00410000000000000	1
0.00420000000000000	1
0.00430000000000000	1
0.00440000000000000	1
0.00450000000000000	1
0.00460000000000000	1
0.00470000000000000	1
0.00480000000000000	1
0.00490000000000000	1
0.00500000000000000	1
0.00510000000000000	1
0.00520000000000000	1
0.00530000000000000	1
0.00540000000000000	1
0.00550000000000000	1
0.00560000000000000	1
0.00570000000000000	1
0.00580000000000000	1
0.00590000000000000	1
0.00600000000000000	1
0.00610000000000000	1
0.00620000000000000	1
0.00630000000000000	1
0.00640000000000000	1
0.00650000000000000	1
0.00660000000000000	1
0.00670000000000000	1
0.00680000000000000	1
0.00690000000000000	1
0.00700000000000000	1
0.00710000000000000	1
0.00720000000000000	1
0.00730000000000000	1
0.00740000000000000	1
0.00750000000000000	1
0.00760000000000000	1
0.00770000000000000	1
0.00780000000000000	1
0.00790000000000000	1
0.00800000000000000	1
0.00810000000000000	1
0.00820000000000000	1
0.00830000000000000	1
0.00840000000000000	1
0.00850000000000000	1
0.00860000000000000	1
0.00870000000000000	1
0.00880000000000000	1
0.00890000000000000	1
0.00900000000000000	1
0.00910000000000000	1
0.00920000000000000	1
0.00930000000000000	1
0.00940000000000000	1
0.00950000000000000	1
0.00960000000000000	1
0.00970000000000000	1
0.00980000000000000	1
0.00990000000000000	1
0.0100000000000000	1
};
\addlegendentry{Case 1, $f_R=5$ GHz}

\addplot[ultra thick, color6]
table {%
0.0001	1
0.0002	1
0.0003	1
0.0004	1
0.0005	1
0.0006	1
0.0007	1
0.0008	1
0.0009	1
0.001	1
0.0011	1
0.0012	1
0.0013	1
0.0014	1
0.0015	1
0.0016	1
0.0017	1
0.0018	1
0.0019	1
0.002	1
0.0021	1
0.0022	1
0.0023	1
0.0024	1
0.0025	1
0.0026	1
0.0027	1
0.0028	1
0.0029	1
0.003	1
0.0031	1
0.0032	1
0.0033	1
0.0034	1
0.0035	1
0.0036	1
0.0037	1
0.0038	1
0.0039	1
0.004	1
0.0041	1
0.0042	1
0.0043	1
0.0044	1
0.0045	1
0.0046	1
0.0047	1
0.0048	1
0.0049	1
0.005	1
0.0051	1
0.0052	1
0.0053	1
0.0054	1
0.0055	1
0.0056	1
0.0057	1
0.0058	1
0.0059	1
0.006	1
0.0061	1
0.0062	1
0.0063	1
0.0064	1
0.0065	1
0.0066	1
0.0067	1
0.0068	1
0.0069	1
0.007	1
0.0071	1
0.0072	1
0.0073	1
0.0074	1
0.0075	1
0.0076	1
0.0077	1
0.0078	1
0.0079	1
0.008	1
0.0081	1
0.0082	1
0.0083	1
0.0084	1
0.0085	1
0.0086	1
0.0087	1
0.0088	1
0.0089	1
0.009	1
0.0091	1
0.0092	1
0.0093	1
0.0094	1
0.0095	1
0.0096	1
0.0097	1
0.0098	1
0.0099	1
0.01	1
};
\addlegendentry{Case 2}

\end{axis}

\end{tikzpicture}
\vspace{-1.05cm}

%% file: tikz_figs/topt_avg_eps.tex
\begin{tikzpicture}

\begin{axis}[
width=\sfwidth,
height=\sfheight,
xmode=log,
legend cell align={left},
legend style={at={(0.995,0.99)}, anchor = north east, fill opacity=0.8, legend columns=1, draw opacity=1, text opacity=1, font=\tiny, draw=white!80!black},
tick align=outside,
tick pos=left,
x grid style={white!69.0196078431373!black},
ylabel={Optimal closed-loop latency  (ms)},
ymin=0, ymax=0.5,
ytick={0,0.1,0.2,0.3,0.4,0.5},
yticklabels={0,100,200,300,400,500},
xtick style={color=black},
y grid style={white!69.0196078431373!black},
xlabel={$\epsilon$},
xmin=1e-4, xmax=1e-2,
ytick style={color=black},
xmajorgrids,
ymajorgrids
]

\addplot [ultra thick, color2, dashed]
table {%
0.000100000000000000	0.378000000000000
0.000200000000000000	0.241000000000000
0.000300000000000000	0.194000000000000
0.000400000000000000	0.170000000000000
0.000500000000000000	0.155000000000000
0.000600000000000000	0.145000000000000
0.000700000000000000	0.138000000000000
0.000800000000000000	0.133000000000000
0.000900000000000000	0.128000000000000
0.00100000000000000	0.125000000000000
0.00110000000000000	0.122000000000000
0.00120000000000000	0.119000000000000
0.00130000000000000	0.117000000000000
0.00140000000000000	0.115000000000000
0.00150000000000000	0.114000000000000
0.00160000000000000	0.113000000000000
0.00170000000000000	0.111000000000000
0.00180000000000000	0.110000000000000
0.00190000000000000	0.109000000000000
0.00200000000000000	0.107000000000000
0.00210000000000000	0.107000000000000
0.00220000000000000	0.106000000000000
0.00230000000000000	0.106000000000000
0.00240000000000000	0.105000000000000
0.00250000000000000	0.105000000000000
0.00260000000000000	0.103000000000000
0.00270000000000000	0.103000000000000
0.00280000000000000	0.102000000000000
0.00290000000000000	0.102000000000000
0.00300000000000000	0.102000000000000
0.00310000000000000	0.102000000000000
0.00320000000000000	0.101000000000000
0.00330000000000000	0.101000000000000
0.00340000000000000	0.101000000000000
0.00350000000000000	0.100000000000000
0.00360000000000000	0.0990000000000000
0.00370000000000000	0.0990000000000000
0.00380000000000000	0.0990000000000000
0.00390000000000000	0.0990000000000000
0.00400000000000000	0.0990000000000000
0.00410000000000000	0.0980000000000000
0.00420000000000000	0.0980000000000000
0.00430000000000000	0.0980000000000000
0.00440000000000000	0.0980000000000000
0.00450000000000000	0.0980000000000000
0.00460000000000000	0.0980000000000000
0.00470000000000000	0.0970000000000000
0.00480000000000000	0.0970000000000000
0.00490000000000000	0.0970000000000000
0.00500000000000000	0.0970000000000000
0.00510000000000000	0.0970000000000000
0.00520000000000000	0.0970000000000000
0.00530000000000000	0.0970000000000000
0.00540000000000000	0.0970000000000000
0.00550000000000000	0.0970000000000000
0.00560000000000000	0.0950000000000000
0.00570000000000000	0.0950000000000000
0.00580000000000000	0.0950000000000000
0.00590000000000000	0.0950000000000000
0.00600000000000000	0.0950000000000000
0.00610000000000000	0.0950000000000000
0.00620000000000000	0.0950000000000000
0.00630000000000000	0.0950000000000000
0.00640000000000000	0.0950000000000000
0.00650000000000000	0.0950000000000000
0.00660000000000000	0.0950000000000000
0.00670000000000000	0.0940000000000000
0.00680000000000000	0.0940000000000000
0.00690000000000000	0.0940000000000000
0.00700000000000000	0.0940000000000000
0.00710000000000000	0.0940000000000000
0.00720000000000000	0.0940000000000000
0.00730000000000000	0.0940000000000000
0.00740000000000000	0.0940000000000000
0.00750000000000000	0.0940000000000000
0.00760000000000000	0.0940000000000000
0.00770000000000000	0.0940000000000000
0.00780000000000000	0.0940000000000000
0.00790000000000000	0.0940000000000000
0.00800000000000000	0.0940000000000000
0.00810000000000000	0.0940000000000000
0.00820000000000000	0.0940000000000000
0.00830000000000000	0.0930000000000000
0.00840000000000000	0.0930000000000000
0.00850000000000000	0.0930000000000000
0.00860000000000000	0.0930000000000000
0.00870000000000000	0.0930000000000000
0.00880000000000000	0.0930000000000000
0.00890000000000000	0.0930000000000000
0.00900000000000000	0.0930000000000000
0.00910000000000000	0.0930000000000000
0.00920000000000000	0.0930000000000000
0.00930000000000000	0.0930000000000000
0.00940000000000000	0.0930000000000000
0.00950000000000000	0.0930000000000000
0.00960000000000000	0.0930000000000000
0.00970000000000000	0.0930000000000000
0.00980000000000000	0.0930000000000000
0.00990000000000000	0.0930000000000000
0.0100000000000000	0.0930000000000000
};
\addlegendentry{$\varrho_\text{th}=0.95$}

\addplot [ultra thick, color3,dashdotted]
table {%
0.000100000000000000	0.412000000000000
0.000200000000000000	0.275000000000000
0.000300000000000000	0.228000000000000
0.000400000000000000	0.204000000000000
0.000500000000000000	0.189000000000000
0.000600000000000000	0.179000000000000
0.000700000000000000	0.172000000000000
0.000800000000000000	0.166000000000000
0.000900000000000000	0.162000000000000
0.00100000000000000	0.159000000000000
0.00110000000000000	0.155000000000000
0.00120000000000000	0.153000000000000
0.00130000000000000	0.151000000000000
0.00140000000000000	0.149000000000000
0.00150000000000000	0.147000000000000
0.00160000000000000	0.146000000000000
0.00170000000000000	0.144000000000000
0.00180000000000000	0.144000000000000
0.00190000000000000	0.143000000000000
0.00200000000000000	0.141000000000000
0.00210000000000000	0.141000000000000
0.00220000000000000	0.140000000000000
0.00230000000000000	0.139000000000000
0.00240000000000000	0.139000000000000
0.00250000000000000	0.139000000000000
0.00260000000000000	0.137000000000000
0.00270000000000000	0.137000000000000
0.00280000000000000	0.136000000000000
0.00290000000000000	0.136000000000000
0.00300000000000000	0.136000000000000
0.00310000000000000	0.135000000000000
0.00320000000000000	0.135000000000000
0.00330000000000000	0.135000000000000
0.00340000000000000	0.135000000000000
0.00350000000000000	0.134000000000000
0.00360000000000000	0.133000000000000
0.00370000000000000	0.133000000000000
0.00380000000000000	0.133000000000000
0.00390000000000000	0.133000000000000
0.00400000000000000	0.132000000000000
0.00410000000000000	0.132000000000000
0.00420000000000000	0.132000000000000
0.00430000000000000	0.132000000000000
0.00440000000000000	0.132000000000000
0.00450000000000000	0.132000000000000
0.00460000000000000	0.131000000000000
0.00470000000000000	0.131000000000000
0.00480000000000000	0.131000000000000
0.00490000000000000	0.131000000000000
0.00500000000000000	0.131000000000000
0.00510000000000000	0.131000000000000
0.00520000000000000	0.131000000000000
0.00530000000000000	0.130000000000000
0.00540000000000000	0.130000000000000
0.00550000000000000	0.130000000000000
0.00560000000000000	0.129000000000000
0.00570000000000000	0.129000000000000
0.00580000000000000	0.129000000000000
0.00590000000000000	0.129000000000000
0.00600000000000000	0.129000000000000
0.00610000000000000	0.129000000000000
0.00620000000000000	0.129000000000000
0.00630000000000000	0.128000000000000
0.00640000000000000	0.128000000000000
0.00650000000000000	0.128000000000000
0.00660000000000000	0.128000000000000
0.00670000000000000	0.128000000000000
0.00680000000000000	0.128000000000000
0.00690000000000000	0.128000000000000
0.00700000000000000	0.128000000000000
0.00710000000000000	0.128000000000000
0.00720000000000000	0.128000000000000
0.00730000000000000	0.128000000000000
0.00740000000000000	0.128000000000000
0.00750000000000000	0.128000000000000
0.00760000000000000	0.128000000000000
0.00770000000000000	0.128000000000000
0.00780000000000000	0.127000000000000
0.00790000000000000	0.127000000000000
0.00800000000000000	0.127000000000000
0.00810000000000000	0.127000000000000
0.00820000000000000	0.127000000000000
0.00830000000000000	0.127000000000000
0.00840000000000000	0.127000000000000
0.00850000000000000	0.127000000000000
0.00860000000000000	0.127000000000000
0.00870000000000000	0.127000000000000
0.00880000000000000	0.127000000000000
0.00890000000000000	0.127000000000000
0.00900000000000000	0.127000000000000
0.00910000000000000	0.127000000000000
0.00920000000000000	0.127000000000000
0.00930000000000000	0.127000000000000
0.00940000000000000	0.127000000000000
0.00950000000000000	0.127000000000000
0.00960000000000000	0.127000000000000
0.00970000000000000	0.127000000000000
0.00980000000000000	0.127000000000000
0.00990000000000000	0.127000000000000
0.0100000000000000	0.126000000000000
};
\addlegendentry{$\varrho_\text{th}=0.99$}

\addplot [ultra thick, color5,densely dotted]
table {%
0.000100000000000000	0.460000000000000
0.000200000000000000	0.323000000000000
0.000300000000000000	0.276000000000000
0.000400000000000000	0.251000000000000
0.000500000000000000	0.237000000000000
0.000600000000000000	0.226000000000000
0.000700000000000000	0.219000000000000
0.000800000000000000	0.214000000000000
0.000900000000000000	0.210000000000000
0.00100000000000000	0.206000000000000
0.00110000000000000	0.203000000000000
0.00120000000000000	0.200000000000000
0.00130000000000000	0.199000000000000
0.00140000000000000	0.196000000000000
0.00150000000000000	0.195000000000000
0.00160000000000000	0.194000000000000
0.00170000000000000	0.192000000000000
0.00180000000000000	0.191000000000000
0.00190000000000000	0.191000000000000
0.00200000000000000	0.189000000000000
0.00210000000000000	0.188000000000000
0.00220000000000000	0.188000000000000
0.00230000000000000	0.187000000000000
0.00240000000000000	0.187000000000000
0.00250000000000000	0.186000000000000
0.00260000000000000	0.185000000000000
0.00270000000000000	0.184000000000000
0.00280000000000000	0.184000000000000
0.00290000000000000	0.184000000000000
0.00300000000000000	0.183000000000000
0.00310000000000000	0.183000000000000
0.00320000000000000	0.183000000000000
0.00330000000000000	0.183000000000000
0.00340000000000000	0.182000000000000
0.00350000000000000	0.182000000000000
0.00360000000000000	0.181000000000000
0.00370000000000000	0.181000000000000
0.00380000000000000	0.180000000000000
0.00390000000000000	0.180000000000000
0.00400000000000000	0.180000000000000
0.00410000000000000	0.180000000000000
0.00420000000000000	0.180000000000000
0.00430000000000000	0.179000000000000
0.00440000000000000	0.179000000000000
0.00450000000000000	0.179000000000000
0.00460000000000000	0.179000000000000
0.00470000000000000	0.179000000000000
0.00480000000000000	0.179000000000000
0.00490000000000000	0.179000000000000
0.00500000000000000	0.179000000000000
0.00510000000000000	0.178000000000000
0.00520000000000000	0.178000000000000
0.00530000000000000	0.178000000000000
0.00540000000000000	0.178000000000000
0.00550000000000000	0.178000000000000
0.00560000000000000	0.177000000000000
0.00570000000000000	0.177000000000000
0.00580000000000000	0.177000000000000
0.00590000000000000	0.177000000000000
0.00600000000000000	0.176000000000000
0.00610000000000000	0.176000000000000
0.00620000000000000	0.176000000000000
0.00630000000000000	0.176000000000000
0.00640000000000000	0.176000000000000
0.00650000000000000	0.176000000000000
0.00660000000000000	0.176000000000000
0.00670000000000000	0.176000000000000
0.00680000000000000	0.176000000000000
0.00690000000000000	0.176000000000000
0.00700000000000000	0.176000000000000
0.00710000000000000	0.176000000000000
0.00720000000000000	0.176000000000000
0.00730000000000000	0.176000000000000
0.00740000000000000	0.175000000000000
0.00750000000000000	0.175000000000000
0.00760000000000000	0.175000000000000
0.00770000000000000	0.175000000000000
0.00780000000000000	0.175000000000000
0.00790000000000000	0.175000000000000
0.00800000000000000	0.175000000000000
0.00810000000000000	0.175000000000000
0.00820000000000000	0.175000000000000
0.00830000000000000	0.175000000000000
0.00840000000000000	0.175000000000000
0.00850000000000000	0.175000000000000
0.00860000000000000	0.175000000000000
0.00870000000000000	0.175000000000000
0.00880000000000000	0.175000000000000
0.00890000000000000	0.175000000000000
0.00900000000000000	0.175000000000000
0.00910000000000000	0.175000000000000
0.00920000000000000	0.175000000000000
0.00930000000000000	0.175000000000000
0.00940000000000000	0.174000000000000
0.00950000000000000	0.174000000000000
0.00960000000000000	0.174000000000000
0.00970000000000000	0.174000000000000
0.00980000000000000	0.174000000000000
0.00990000000000000	0.174000000000000
0.0100000000000000	0.174000000000000
};
\addlegendentry{$\varrho_\text{th}=0.999$}

\end{axis}

\end{tikzpicture}
\vspace{-0.3cm}

%% file: tikz_figs/topt_avg_rho.tex
\begin{tikzpicture}

\begin{axis}[
width=\sfwidth,
height=\sfheight,
legend cell align={left},
legend style={at={(0.005,0.99)}, anchor = north west, fill opacity=0.8, legend columns=1, draw opacity=1, text opacity=1, font=\tiny, draw=white!80!black},
tick align=outside,
tick pos=left, 
x grid style={white!69.0196078431373!black},
ylabel={Optimal closed-loop latency  (ms)},
ymin=0.3, ymax=0.36,
ytick={0.3,0.320,0.340,0.360},
yticklabels={300,320,340,360},
xtick style={color=black},
y grid style={white!69.0196078431373!black},
xlabel={$\varrho_\text{th}$},
xmin=0.5, xmax=1,
ytick style={color=black},
xmajorgrids,
ymajorgrids
]

\addplot [ultra thick, color2, dashed]
table {%
0.500000000000000	0.317000000000000
0.510000000000000	0.318000000000000
0.520000000000000	0.318000000000000
0.530000000000000	0.319000000000000
0.540000000000000	0.320000000000000
0.550000000000000	0.321000000000000
0.560000000000000	0.321000000000000
0.570000000000000	0.322000000000000
0.580000000000000	0.323000000000000
0.590000000000000	0.324000000000000
0.600000000000000	0.325000000000000
0.610000000000000	0.326000000000000
0.620000000000000	0.326000000000000
0.630000000000000	0.327000000000000
0.640000000000000	0.328000000000000
0.650000000000000	0.329000000000000
0.660000000000000	0.330000000000000
0.670000000000000	0.331000000000000
0.680000000000000	0.332000000000000
0.690000000000000	0.333000000000000
0.700000000000000	0.334000000000000
0.710000000000000	0.335000000000000
0.720000000000000	0.336000000000000
0.730000000000000	0.337000000000000
0.740000000000000	0.338000000000000
0.750000000000000	0.339000000000000
0.760000000000000	0.340000000000000
0.770000000000000	0.341000000000000
0.780000000000000	0.342000000000000
0.790000000000000	0.344000000000000
0.800000000000000	0.345000000000000
0.810000000000000	0.346000000000000
0.820000000000000	0.347000000000000
0.830000000000000	0.349000000000000
0.840000000000000	0.350000000000000
0.850000000000000	0.352000000000000
0.860000000000000	0.354000000000000
0.870000000000000	0.355000000000000
0.880000000000000	0.357000000000000
0.890000000000000	0.359000000000000
0.900000000000000	0.362000000000000
0.910000000000000	0.364000000000000
0.920000000000000	0.367000000000000
0.930000000000000	0.370000000000000
0.940000000000000	0.373000000000000
0.950000000000000	0.378000000000000
0.960000000000000	0.383000000000000
0.970000000000000	0.389000000000000
0.980000000000000	0.398000000000000
0.990000000000000	0.412000000000000
};
\addlegendentry{Statistical sense optimization}

\addplot [ultra thick, color5,densely dotted]
table {
0.510000000000000	0.320440000000000
0.520000000000000	0.320440000000000
0.530000000000000	0.320440000000000
0.540000000000000	0.320440000000000
0.550000000000000	0.320440000000000
0.560000000000000	0.320440000000000
0.570000000000000	0.320440000000000
0.580000000000000	0.320440000000000
0.590000000000000	0.320440000000000
0.600000000000000	0.320440000000000
0.610000000000000	0.320440000000000
0.620000000000000	0.320440000000000
0.630000000000000	0.320440000000000
0.640000000000000	0.320440000000000
0.650000000000000	0.320440000000000
0.660000000000000	0.320440000000000
0.670000000000000	0.320440000000000
0.680000000000000	0.320440000000000
0.690000000000000	0.320440000000000
0.700000000000000	0.320440000000000
0.710000000000000	0.320440000000000
0.720000000000000	0.320440000000000
0.730000000000000	0.320440000000000
0.740000000000000	0.320440000000000
0.750000000000000	0.320440000000000
0.760000000000000	0.320440000000000
0.770000000000000	0.320440000000000
0.780000000000000	0.320440000000000
0.790000000000000	0.320440000000000
0.800000000000000	0.320440000000000
0.810000000000000	0.320440000000000
0.820000000000000	0.320440000000000
0.830000000000000	0.320440000000000
0.840000000000000	0.320440000000000
0.850000000000000	0.320440000000000
0.860000000000000	0.320440000000000
0.870000000000000	0.320440000000000
0.880000000000000	0.320440000000000
0.890000000000000	0.320440000000000
0.900000000000000	0.320440000000000
0.910000000000000	0.320440000000000
0.920000000000000	0.320440000000000
0.930000000000000	0.320440000000000
0.940000000000000	0.320440000000000
0.950000000000000	0.320440000000000
0.960000000000000	0.320440000000000
0.970000000000000	0.320440000000000
0.980000000000000	0.320440000000000
0.990000000000000	0.320440000000000
};
\addlegendentry{Average sense optimization}

\end{axis}

\end{tikzpicture}
\vspace{-0.3cm}

%% file: main.bbl
\begin{thebibliography}{10}
\providecommand{\url}[1]{#1}
\csname url@samestyle\endcsname
\providecommand{\newblock}{\relax}
\providecommand{\bibinfo}[2]{#2}
\providecommand{\BIBentrySTDinterwordspacing}{\spaceskip=0pt\relax}
\providecommand{\BIBentryALTinterwordstretchfactor}{4}
\providecommand{\BIBentryALTinterwordspacing}{\spaceskip=\fontdimen2\font plus
\BIBentryALTinterwordstretchfactor\fontdimen3\font minus
  \fontdimen4\font\relax}
\providecommand{\BIBforeignlanguage}[2]{{%
\expandafter\ifx\csname l@#1\endcsname\relax
\typeout{** WARNING: IEEEtran.bst: No hyphenation pattern has been}%
\typeout{** loaded for the language `#1'. Using the pattern for}%
\typeout{** the default language instead.}%
\else
\language=\csname l@#1\endcsname
\fi
#2}}
\providecommand{\BIBdecl}{\relax}
\BIBdecl

\bibitem{Fettweis_TI}
G.~P. {Fettweis}, ``The tactile internet: Applications and challenges,''
  \emph{IEEE Veh. Technol. Mag.}, vol.~9, no.~1, pp. 64--70, 2014.

\bibitem{EMG_teleoperation_1}
J.~{Luo} \emph{et~al.}, ``A teleoperation framework for mobile robots based on
  shared control,'' \emph{IEEE Robot. Autom. Lett.}, vol.~5, no.~2, pp.
  377--384, 2020.

\bibitem{antonakoglou2018toward}
K.~Antonakoglou \emph{et~al.}, ``Toward haptic communications over the {5G
  Tactile Internet},'' \emph{IEEE Comm. Surv. Tut.}, vol.~20, no.~4, pp.
  3034--3059, 2018.

\bibitem{zhao2017estimating}
J.~Zhao \emph{et~al.}, ``Estimating the motion-to-photon latency in head
  mounted displays,'' in \emph{IEEE Virtual Reality (VR)}.\hskip 1em plus 0.5em
  minus 0.4em\relax IEEE, 2017, pp. 313--314.

\bibitem{transparency_work}
S.~{Hirche} and M.~{Buss}, ``Human-oriented control for haptic teleoperation,''
  \emph{Proc. IEEE}, vol. 100, no.~3, pp. 623--647, 2012.

\bibitem{stability_td_hri}
F.~Müller \emph{et~al.}, ``Stability of nonlinear time-delay systems
  describing human–robot interaction,'' \emph{IEEE/ASME Trans. Mechatronics},
  vol.~24, no.~6, pp. 2696--2705, 2019.

\bibitem{kumcu2017effect}
A.~Kumcu \emph{et~al.}, ``Effect of video lag on laparoscopic surgery:
  correlation between performance and usability at low latencies,'' \emph{The
  Int. J. Med. Robot. Computer Assisted Surgery}, vol.~13, no.~2, p. e1758,
  2017.

\bibitem{barba2022remote}
P.~Barba \emph{et~al.}, ``Remote telesurgery in humans: a systematic review,''
  \emph{Surgical Endoscopy}, pp. 1--7, 2022.

\bibitem{valenzuela2019virtual}
D.~Valenzuela-Urrutia \emph{et~al.}, ``Virtual reality-based time-delayed
  haptic teleoperation using point cloud data,'' \emph{J. Intell. Robot.
  Syst.}, vol.~96, no.~3, pp. 387--400, 2019.

\bibitem{effect_of_delay}
Z.~Shi \emph{et~al.}, ``Effects of packet loss and latency on the temporal
  discrimination of visual-haptic events,'' \emph{IEEE Trans. Haptics}, vol.~3,
  no.~1, pp. 28--36, 2010.

\bibitem{low_latency_h2m}
L.~Ruan, M.~P.~I. Dias, and E.~Wong, ``Achieving low-latency human-to-machine
  (h2m) applications: An understanding of h2m traffic for ai-facilitated
  bandwidth allocation,'' \emph{IEEE Internet Things J.}, vol.~8, no.~1, pp.
  626--635, 2021.

\bibitem{closing_the_force_loop}
R.~Balachandran \emph{et~al.}, ``Closing the force loop to enhance transparency
  in time-delayed teleoperation,'' in \emph{2020 IEEE Int. Conf. Robot. Autom.
  (ICRA)}, 2020, pp. 10\,198--10\,204.

\bibitem{IoT_TI}
X.~{Ge}, R.~{Zhou}, and Q.~{Li}, ``{5G NFV}-based {T}actile internet for
  mission-critical {IoT} services,'' \emph{IEEE Internet Things J.}, vol.~7,
  no.~7, pp. 6150--6163, 2020.

\bibitem{NFV_e2e_QoS}
N.~Gholipoor \emph{et~al.}, ``{E2E QoS} guarantee for the {T}actile internet
  via joint {NFV} and radio resource allocation,'' \emph{IEEE Trans. Netw.
  Service Manag.}, vol.~17, no.~3, pp. 1788--1804, 2020.

\bibitem{onu_2}
M.~Maier and A.~Ebrahimzadeh, ``Towards immersive tactile internet experiences:
  {L}ow-latency {FiWi} enhanced mobile networks with edge intelligence
  [invited],'' \emph{IEEE J. Opt. Commun. Netw.}, vol.~11, no.~4, pp. B10--B25,
  2019.

\bibitem{MEC_tactility}
Y.~{Xiao} and M.~{Krunz}, ``Distributed optimization for energy-efficient fog
  computing in the tactile internet,'' \emph{IEEE J. Sel. Areas Commun.},
  vol.~36, no.~11, pp. 2390--2400, 2018.

\bibitem{MEC_TII}
M.~{Aazam}, K.~A. {Harras}, and S.~{Zeadally}, ``Fog computing for {5G} tactile
  industrial {I}nternet of {T}hings: {QoE}-aware resource allocation model,''
  \emph{IEEE Trans. Ind. Informat.}, vol.~15, no.~5, pp. 3085--3092, 2019.

\bibitem{hybrid_caching_TI}
J.~Xu, K.~Ota, and M.~Dong, ``Energy efficient hybrid edge caching scheme for
  tactile internet in 5{G},'' \emph{IEEE Trans. Green Commun. Netw.}, vol.~3,
  no.~2, pp. 483--493, 2019.

\bibitem{TelSurg}
S.~Sedaghat and A.~H. Jahangir, ``{RT-TelSurg: R}eal time telesurgery using
  {SDN}, fog, and cloud as infrastructures,'' \emph{IEEE Access}, vol.~9, pp.
  52\,238--52\,251, 2021.

\bibitem{MEC_survey}
Y.~Mao \emph{et~al.}, ``A survey on mobile edge computing: {T}he communication
  perspective,'' \emph{IEEE Commun. Surveys Tuts.}, vol.~19, no.~4, pp.
  2322--2358, 2017.

\bibitem{MEC_survey_IoTJ}
N.~Abbas \emph{et~al.}, ``Mobile edge computing: {A} survey,'' \emph{IEEE
  Internet Things J.}, vol.~5, no.~1, pp. 450--465, 2018.

\bibitem{our_icc}
S.~Suman \emph{et~al.}, ``Analysis and optimization of the latency budget in
  wireless systems with mobile edge computing,'' in \emph{IEEE Int. Conf.
  Commun.}, 2022, pp. 1--6 (accepted).

\bibitem{o-ran}
``{O-RAN} near-real-time {RAN} intelligent controller architecture \& {E2}
  general aspects and principles – v1.01,'' Tech. Spec., 2020.

\bibitem{cyber_glove}
\BIBentryALTinterwordspacing
Cyberglove {S}ystems {I}nc. (2017). [Online]. Available:
  \url{http://www.cyberglovesystems.com/cybergrasp/}
\BIBentrySTDinterwordspacing

\bibitem{s_dosen_ET_fb}
J.~L. {Dideriksen}, I.~U. {Mercader}, and S.~{Dosen}, ``Closed-loop control
  using electrotactile feedback encoded in frequency and pulse width,''
  \emph{IEEE Trans. Haptics}, vol.~13, no.~4, pp. 818--824, 2020.

\bibitem{onu_1}
M.~Chowdhury and M.~Maier, ``Local and nonlocal human-to-robot task allocation
  in fiber-wireless multi-robot networks,'' \emph{IEEE Syst. J.}, vol.~12,
  no.~3, pp. 2250--2260, 2018.

\bibitem{onu_3}
A.~Ebrahimzadeh and M.~Maier, ``Delay-constrained teleoperation task scheduling
  and assignment for human+machine hybrid activities over {FiWi} enhanced
  networks,'' \emph{IEEE Trans. Netw. Service Manag.}, vol.~16, no.~4, pp.
  1840--1854, 2019.

\bibitem{coll_computing}
M.~Chowdhury and M.~Maier, ``Collaborative computing for advanced tactile
  internet human-to-robot {(H2R)} communications in integrated {FiWi}
  multirobot infrastructures,'' \emph{IEEE Internet Things J.}, vol.~4, no.~6,
  pp. 2142--2158, 2017.

\bibitem{cross_layer_3}
C.~She, C.~Yang, and T.~Q.~S. Quek, ``Cross-layer transmission design for
  tactile internet,'' in \emph{IEEE GLOBECOM}, 2016, pp. 1--6.

\bibitem{TI_channel_access}
Y.~Feng \emph{et~al.}, ``Hybrid coordination function controlled channel access
  for latency-sensitive tactile applications,'' in \emph{IEEE GLOBECOM}, 2017,
  pp. 1--6.

\bibitem{cross_layer_1}
N.~Gholipoor, H.~Saeedi, and N.~Mokari, ``Cross-layer resource allocation for
  mixed tactile internet and traditional data in scma based wireless
  networks,'' in \emph{IEEE Wireless Commun. Netw. Conf. Wksp. (WCNCW)}, 2018,
  pp. 356--361.

\bibitem{NFV_JSAC}
Z.~Xiang \emph{et~al.}, ``Reducing latency in virtual machines: Enabling
  tactile internet for human-machine co-working,'' \emph{IEEE J. Sel. Areas
  Commun.}, vol.~37, no.~5, pp. 1098--1116, 2019.

\bibitem{aijaz2017shaping}
A.~Aijaz \emph{et~al.}, ``{Shaping 5G for the Tactile Internet},'' in \emph{5G
  Mobile Commun.}\hskip 1em plus 0.5em minus 0.4em\relax Springer, 2017, pp.
  677--691.

\bibitem{goldsmith2005wireless}
A.~Goldsmith, \emph{Wireless communications}.\hskip 1em plus 0.5em minus
  0.4em\relax Cambridge university press, 2005.

\bibitem{robot_comput}
G.~Hu, W.~P. Tay, and Y.~Wen, ``Cloud robotics: {A}rchitecture, challenges and
  applications,'' \emph{IEEE Netw.}, vol.~26, no.~3, pp. 21--28, 2012.

\bibitem{MEC_random_var}
W.~Yuan and K.~Nahrstedt, ``Energy-efficient soft real-time {CPU} scheduling
  for mobile multimedia systems,'' \emph{SIGOPS Oper. Syst. Rev.}, vol.~37,
  no.~5, p. 149–163, Oct. 2003.

\bibitem{MEC_gamma_1}
D.~Han \emph{et~al.}, ``Offloading optimization and bottleneck analysis for
  mobile cloud computing,'' \emph{IEEE Trans. Commun.}, vol.~67, no.~9, pp.
  6153--6167, 2019.

\bibitem{MEC_gamma_2}
S.~Jošilo and G.~Dán, ``Selfish decentralized computation offloading for
  mobile cloud computing in dense wireless networks,'' \emph{IEEE Trans. Mobile
  Comput.}, vol.~18, no.~1, pp. 207--220, 2019.

\bibitem{jsac_comp_model}
X.~Li \emph{et~al.}, ``Wirelessly powered crowd sensing: {J}oint power
  transfer, sensing, compression, and transmission,'' \emph{IEEE J. Sel. Areas
  Commun.}, vol.~37, no.~2, pp. 391--406, 2019.

\bibitem{data_comp_decomp_1}
R.~Kothiyal \emph{et~al.}, ``Energy and performance evaluation of lossless file
  data compression on server systems,'' in \emph{Proc. of SYSTOR}, ser. SYSTOR
  '09.\hskip 1em plus 0.5em minus 0.4em\relax New York, USA: Association for
  Computing Machinery, 2009.

\bibitem{data_comp_decomp_4}
J.~Ren, Y.~Ruan, and G.~Yu, ``Data transmission in mobile edge networks:
  Whether and where to compress?'' \emph{IEEE Commun. Lett.}, vol.~23, no.~3,
  pp. 490--493, 2019.

\bibitem{comm_lett_comp}
J.-B. Wang \emph{et~al.}, ``Joint optimization of transmission bandwidth
  allocation and data compression for mobile-edge computing systems,''
  \emph{IEEE Commun. Lett.}, vol.~24, no.~10, pp. 2245--2249, 2020.

\bibitem{data_comp_decomp_2}
M.~Burrows \emph{et~al.}, ``On-line data compression in a log-structured file
  system,'' in \emph{Proc. Int. Conf. Architectural Support Programming
  Languages Operating Syst.}, ser. ASPLOS V.\hskip 1em plus 0.5em minus
  0.4em\relax New York, NY, USA: Association for Computing Machinery, 1992, p.
  2–9.

\bibitem{data_comp_decomp_3}
Z.~Zhang \emph{et~al.}, ``Efficient {I/O} for neural network training with
  compressed data,'' in \emph{IEEE Int. Parallel Distrib. Process. Symp.
  (IPDPS)}, 2020, pp. 409--418.

\bibitem{gamma_pdf}
H.~C. Thom, ``A note on the gamma distribution,'' \emph{Monthly Weather
  Review}, vol.~86, no.~4, pp. 117--122, 1958.

\bibitem{gamma_approx_1}
P.~G. Moschopoulos, ``The distribution of the sum of independent gamma random
  variables,'' \emph{Annals of the Institute of Statistical Mathematics},
  vol.~37, no.~3, pp. 541--544, 1985.

\bibitem{gamma_approx_2}
H.~Murakami, ``Approximations to the distribution of sum of independent
  non-identically gamma random variables,'' \emph{Mathematical Sciences},
  vol.~9, no.~4, pp. 205--213, 2015.

\bibitem{gamma_approx_3}
L.~Tlebaldiyeva, B.~Maham, and T.~A. Tsiftsis, ``Capacity analysis of
  device-to-device mmwave networks under transceiver distortion noise and
  imperfect csi,'' \emph{IEEE Trans. Veh. Technol.}, vol.~69, no.~5, pp.
  5707--5712, 2020.

\bibitem{halawa2017nvidia}
H.~Halawa \emph{et~al.}, ``{NVIDIA Jetson} platform characterization,'' in
  \emph{Proc. European Conf. Parallel Process.}\hskip 1em plus 0.5em minus
  0.4em\relax Springer, 2017, pp. 92--105.

\bibitem{log_concave_1}
M.~Bagnoli and T.~Bergstrom, ``Log-concave probability and its applications,''
  \emph{Economic Theory}, vol.~26, no.~2, pp. 445--469, 2005.

\bibitem{unimodal_convex}
S.~Dharmadhikari and K.~Joag-Dev, \emph{Unimodality, Convexity, and
  Applications}.\hskip 1em plus 0.5em minus 0.4em\relax Elsevier, 1988.

\end{thebibliography}
